\documentclass[manuscript]{geophysics}
\usepackage{amsmath,amscd,amssymb}
\usepackage{amsthm}
\usepackage{mathrsfs}
\newcommand{\bx}{{\bf x}}
\newcommand{\bff}{{\bf f}}
\newcommand{\bR}{{\bf R}}
\newcommand{\bN}{{\bf N}}

\newcommand{\half}{\ensuremath{\frac{1}{2}}}

\newcommand{\bv}{\mathbf{v}}
\newcommand{\bw}{\mathbf{w}}

\newcommand{\ee}{\eta_n}

\newcommand{\lh}{{L^2 ({\bf R}, H)}}
\newcommand{\uml}{u_{m(l)}}

\newcommand{\llh}{L^2_{\rm loc}({\bf R},H)}
\newcommand{\lch}{L^2_{\rm caus}({\bf R},H)}
\newcommand{\lph}{L^2_{+}({\bf R},H)}

\newtheorem{theorem}{Theorem}
\newtheorem{lemma}{Lemma}
\newtheorem{corollary}{Corollary}
\newtheorem{proposition}{Proposition}

\allowdisplaybreaks[1]

\title{A Mathematical Framework for Inverse Wave Problems in Heterogeneous Media}
\author{Kirk D. Blazek\footnotemark[3], Christiaan Stolk\footnotemark[2], and William W. Symes\footnotemark[1]}
\address{
\footnotemark[1]
The Rice Inversion Project,
Department of Computational and Applied Mathematics, Rice University,
Houston TX 77251-1892 USA, email {\tt symes@caam.rice.edu}\\
\footnotemark[2]
KdV Institute for Mathematics,
University of Amsterdam, P. O. Box 94216, 1090 GE Amsterdam, The Netherlands,
email {\tt c.c.stolk@uva.nl}\\
\footnotemark[3]
Rochester Community and Technical College,
Rochester, MN 55904 USA, email {\tt kirk.blazek@gmail.com}

}

\righthead{Inverse Problems for Hyperbolic Systems}
\lefthead{Blazek, Stolk, and Symes}

\begin{document}

\setlength{\baselineskip}{20pt plus1pt}

\maketitle

\begin{abstract}
This paper provides a theoretical foundation for some common formulations of inverse problems in wave propagation, based on hyperbolic systems of linear integro-differential equations with bounded and measurable coefficients. 
The coefficients of these time-dependent partial differential equations respresent parametrically the spatially varying mechanical properties of materials. 
Rocks, manufactured materials, and other wave propagation environments often exhibit spatial heterogeneity in mechanical properties at a wide variety of scales, and coefficient functions representing these properties must mimic this heterogeneity. We show how to choose domains (classes of nonsmooth coefficient functions) and data definitions (traces of weak solutions) so that optimization formulations of inverse wave problems satisfy some of the prerequisites for application of Newton's method and its relatives. These results follow from the properties of a class of abstract first-order evolution systems, of which various physical wave systems appear as concrete instances. Finite speed of propagation for linear waves with bounded, measurable mechanical parameter fields is one of the by-products of this theory.

\end{abstract}

\section{Introduction}

Inverse problems for waves in heterogeneous materials occur in seismology, ultrasonic nondestructive evaluation and biomedical imaging, some electromagnetic imaging technologies, and elsewhere. The data of such problems are often idealized as traces of wave fields on space-time hypersurfaces (or surfaces of even lower dimension), and nonlinear least-squares (or other data-fitting) methods developed for their numerical solution, that is, for finding the coefficients in the systems of partial differential equations chosen to model the wave physics. Newton's method (or one of its relatives) is a natural choice of solution method for these optimization formulations, as its rapid convergence compensates to some extent for the very large computational size of reasonable discretizations - see for example \cite{Ghattas:IP25}. 

Natural and man-made materials may be so heterogeneous that few limits can be placed on the spatial regularity of coefficients representing material parameters (for the case of sedimentary rocks, see  for instance \cite[]{WaldenHosken:86,BouCouZin:87,WhiteShengNair:90}). However, traces are well-defined only for fields possessing some space-time regularity, and optimal convergence of Newton-like methods requires that objective functions and constraints be regular functions of the model parameters, in some sense. Since nonsmoothness of the material parameter (coefficient) fields implies nonsmoothness of dynamical (solution) fields, it is {\em a priori} unclear that the required traces exist, or that the fields depend sufficiently smoothly on the material parameter (coefficient) functions to permit Newton-like optimization methods to be applied. While discretized models may {\em a fortiori} produce nominally well-defined data simulations and smooth objective (data misfit) functions, accurate approximation of the continuum fields implies that the continuum limit will determine the behavior of solution algorithms for the corresponding discretized inverse problems. Also, numerical solutions to discretized inverse problems respect the underlying continuum physics only insofar as they converge under refinement of the discretization (at least in principle), and such convergence generally requires that the continuum problems have well-behaved solutions.  

The theory presented here provides a mathematical foundation for some common idealizations of inverse problems in wave propagation and approaches to their solution, for material models of bounded and measurable dependence on space variables. Measurability means roughly the existence of averages over arbitrary volumes, which seems a reasonable requirement for physical parameter fields. Most such fields are bounded as a matter of principle (positivity of density) or observation (p-wave velocity range in metals,...).

The dynamical laws underlying such problems are symmetric hyperbolic systems, either differential or integrodifferential, for a vector-valued fields $u$ defined on a domain $\Omega \subset \bR^d$ (of course, typically $d=1,2$ or $3$). Denoting the dimension of the dynamical state vector $u$ by $k \in \bN$, the systems treated in this paper take the form
\begin{equation}
\label{symmhyp}
a \frac{\partial u}{\partial t} + p(\nabla) u + bu + q \ast u = f\mbox{ in }\Omega \times \bR; \, u = 0 \mbox{ for } t < 0,
\end{equation}
in which $a, b$ are $k \times k$ matrix-valued functions on $\Omega$, $q$ is a $k \times k$ matrix-valued function on $\Omega \times \bR$ with $q(\bx,t)=0$ for $t<0$, and $p(\nabla)$ is a $k \times k$ matrix of constant-coefficient first-order differential operators in the space variables:
\begin{equation}\label{def-D}
p(\nabla)u=\sum_{j=1}^{d} p_j \frac{\partial u}{\partial x_j}, \,\,p_j \in \bR^{k \times k}, \,i=1,...,d.
\end{equation}
The right-hand side vector $f$ is a $k$-vector valued function on $\Omega \times \bR$, vanishing for $t<0$ and representing energy input to the system.

Such a system is {\em symmetric hyperbolic} if $a(\bx), q(\bx,t)$ and $p_j$ are symmetric for all $\bx \in \Omega, t \in \bR$, respectively $j=1,...,d$, and $a(\bx)$ is uniformly positive-definite: $C_*, C^* \in \bR$ exist so that $0 < C_* \le C^*$, and
\begin{equation}
\label{elliptic}
C_*I \le a(\bx) \le C^*I, \,\, \mbox{a. e. }\bx \in \Omega
\end{equation}
 
This paper has three central objectives. First, we describe constraints on data (coefficients $a,b,p,q$ and right-hand side $f$) and solution under which \eqref{symmhyp} is well-posed, even when the coefficient functions $a,b,q$ are permitted to be quite irregular in their dependence on $\bx \in \Omega$ - in fact, merely bounded and measureable. Second, we determine a sense in which the solution of \eqref{symmhyp} is regular ($C^k$, $k \ge 0$) as a function of its coefficients. Third, we characterize conditions under which traces of solutions on time-like hypersurfaces are well-defined, and regular as functions of the coefficients.

As in the similar study of \cite{stolk} for second-order hyperbolic systems, we represent \eqref{symmhyp} as an instance of a class of abstract autonomous integro-differential systems.  We study these abstract systems within a framework modeled after the duality approach introduced by J.-L. Lions and his collaborators in the 1960's \cite[]{Lions:71,LionsMagenes:72}. A natural notion of {\em weak solution} is central to this approach, and complements that of strong (pointwise) solution. The class of first-order systems studied here possesses a remarkable property: smoothing solutions in time converts weak solutions to strong (pointwise in $t$) solutions. Specialized to systems such as acoustics or elastodynamics, this property implies that appropriate traces of solutions are well-defined, provided that the right-hand-side or source term is minimally smooth in time (only!). Existence, uniqueness, and regularity of the solution as function of problem data - both right-hand side and (operator) coefficients - also follows from this fact. 

The  abstract theory accomplishes even more than that, as it applies to a much wider range of dynamics than those commonly occurring in continuum mechanics, accommodating first-order systems with operator coefficients. This added generality creates no additional difficulty for the basic theory. It is in fact very useful: as will be explained in the Discussion section, it justifies certain infeasible-model methods for inverse problems in wave propagation.

Another useful by-product of the theory is the finite speed of propagation property for hyperbolic systems with bounded, measurable coefficients, a result which so far as we can tell is new. This follows from the similar property for systems with smooth coefficients via the continuous dependence of solutions on the coefficients in the sense of convergence in measure. 

Precise statements of the main results for symmetric hyperbolic systems \eqref{symmhyp} appear in the next section, followed by a brief discussion of the related literature. The third section contains the definition of the class of abstract systems studied in this paper, and statements of the main results to be established concerning it. Proofs of these results follow in the next three sections. The seventh section shows how symmetric hyperbolic systems fit into the abstract framework, and contains proofs of the main theorems stated in the second section. The eighth section treats the hyperbolic case of linear viscoelasticity as an instance of the theory developed earlier. We end with a discussion of the potential importance of the general (operator coefficient) case of the abstract theory in formulating certain inversion algorithms, representation of energy sources, and several other matters not addressed in the body of the paper. Two appendices deal with the skew-adjointness of the acoustic grad-div operator, and the case of systems with no memory term ($q=0$ in \eqref{symmhyp}), for which initial value problems make sense.

\section{Main Results for Symmetric Hyperbolic Systems}
Our results hinge on an assumption concerning the skew-symmetric differential operator $p(\nabla)$, which to begin with may be viewed as a densely-defined operator on $L^2(\Omega)^k$, whose domain is a subset of $ C^{\infty}(\Omega^{\rm int})^k$  that contains $C^{\infty}_0(\Omega^{\rm int})^k$. We assume that $p(\nabla)$ extends to a skew-adjoint operator operator $P$ with dense domain $V \subset L^2(\Omega)^k$, and that $V$ is metrized with the graph norm of $p(\nabla)$ or its equivalent, which we denote by $\| \cdot \|_V$.

With this assumption, the first of our two major sets of results on symmetric hyperbolic systems addresses well-posedness. The conclusion with the most direct implications for the formulation of inverse problems is 

\begin{theorem}\label{symmhyp:wellposed}
Suppose that, in addition to the hypotheses explained above, the right-hand side $f$ has a square-integrable $t$-derivative: $f \in H^1_{\rm loc}(\bR,L^2(\Omega)^k)$, and is {\em causal}, that is, $f = 0, t < 0$. Then there exists a unique causal $u \in C^1(\bR,L^2(\Omega)^k) \cap C^0(\bR,V)$ satisfying \eqref{symmhyp} a. e. in $\Omega$ for each $t \in \bR$. Moreover, there exists an increasing function $C: \bR_+ \rightarrow \bR_+$, depending on $C_*, C^*$, and bounds for $b,q$, so that for every $t \in \bR$,
\begin{equation}
\label{symmhyp:energy}
\left\|\frac{\partial u}{\partial t}(\cdot,t)\right\|^2_{L^2(\Omega)^k} + \|u(\cdot,t)\|^2_V \le C(t) \int_0^t d\tau 
\left[ 
\|f(\cdot,\tau)\|^2_{L^2(\Omega)^k}
+\left\|\frac{\partial u}{\partial t}(\cdot,\tau)\right\|^2_{L^2(\Omega)^k} 
\right]
\end{equation}
\end{theorem}

The second main result addresses the existence of traces and their regular dependence on the coefficients.
Suppose that $\Gamma \subset \Omega$ is a smoothly embedded $m$-dimensional submanifold, $m<d$. Denote by $\Sigma_{\Gamma}$ the trace map  $C^{\infty}_0(\Omega^{\rm int})^k \rightarrow L^2(\Gamma)^k$. Suppose that $l \in \bN$ and $m \in C_0^{\infty}(\Gamma,\bR^{l\times k})$. Define $S[m]:C^{\infty}_0(\Omega^{\rm int})^k \rightarrow L^2(\Gamma)^l$ by
\[
S[m]u = m \Sigma_{\Gamma}u.
\]
Note that $S[m]$ induces a map $L^2([0,T],C^{\infty}_0(\Omega^{\rm int})^k) \rightarrow L^2([0,T],L^2(\Gamma)^l)$, for any $T>0$. Abusing notation, call this induced map $S[m]$ as well.

Having parametrized the various possible trace maps, we define the data prediction, or forward, operator, ${\cal F}_{f,m}$, by
\[
{\cal F}_{f,m}(a,b,q) = S[m]u
\]
in which $u$ is the causal solution of \eqref{symmhyp} provided by Theorem \ref{symmhyp:wellposed}.
A suitable domain for ${\cal F}_{f,m}$ is $M \subset L^{\infty}(\Omega,\bR^{k \times k}_{\rm symm}) \times L^{\infty}(\Omega, \bR^{k\times k} ) \times L^1(\bR_+,L^{\infty}(\Omega,\bR^{k \times k}_{\rm symm} ))$, defined by 
\[
M(C_*,C^*,C_B,C_Q) = \{(a,b,q): C_*I < a(\bx) < C^*I \mbox{ for all }\bx \in \Omega, \|b\|_{L^{\infty}(\Omega, \bR^{k\times k} )} < C_B, 
\]
\[
\|q\|_{L^1(\bR_+,L^{\infty}(\Omega,\bR^{k \times k}_{\rm symm} ))} < C_Q, q(t) = 0 \mbox{ for }t<0\},
\]
\[
M = \bigcup \{ M(C_*,C^*,C_B,C_Q): 0 < C_* \le C^*, C_B, C_Q \in \bR_+\}.
\]

\begin{theorem}\label{symmhyp:tracereg} Suppose that (i) $S[m]$ extends continuously to $V$, and (ii) $s \ge 2$ and $f \in H^s(\bR,L^2(\Omega)^k)$. Then for any $T > 0$,  ${\cal F}_{f,m}:  M \rightarrow C^0([0,T],L^2(\Gamma)^l)$ is well-defined and of class $C^{s-1}$.
\end{theorem}

A stronger result concerning continuity {\em per se}, involving a weaker topology on the coefficient set $M$, is also important. Recall that a sequence of real-valued measureable functions $\{g_m: m \in \bN\} \subset L^{\infty}(\Omega)$ {\em converges to zero in measure} iff for any $\epsilon > 0$, 
\[
e_m = |\{\bx: |g_m(\bx)| > \epsilon\}|
\]
is a null sequence.

\begin{theorem}\label{symmhyp:strong} Suppose that $0 < C_* \le C^*, 0 \le C_B, C_Q$, 
$\{(a_m,b_m,q_m): m\in \bN\} \subset M(C_*,C^*,C_B,C_Q)$, and $(a,b,q) \in (C_*,C^*,C_B,C_Q)$. Suppose further that
\begin{itemize}
\item[i. ]$a_m \rightarrow a$ in measure,
\item[ii. ]$b_m \rightarrow b$ in measure,
\item[iii. ]$\int_0^t|q_m -q| \rightarrow 0$ in measure 
\end{itemize}
Denote by $u_m$ the (strong) solution of \eqref{symmhyp} provided by Theorem 1, with coefficients $a_m,b_m,$ and $q_m$, likewise by $u$ the solution with coefficients $a, b,$
and $q$, with common causal right-hand side $f \in H^1_{\rm loc}(\bR,L^2(\Omega)^k)$. Then for any $T \in \bR_+$,
\[
\lim_{m \rightarrow \infty} \|u_m-u\|_{L^{\infty}([0,T],L^2(\Omega))} = 0.
\]
\end{theorem}

One way of generating sequences converging in measure is by mollification. We exploit this observation to use Theorem \ref{symmhyp:strong} in conjunction with a well-known fact about symmetric hyperbolic systems with smooth coefficients, to prove

\begin{theorem}\label{symmhyp:speed} Suppose that $(a,b,q) \in M$, $\epsilon > 0$, and that $\tau \in \bR_+$  satisfies
\begin{equation}\label{tau-bd}
\tau a(\bx)+\sum_{i=1}^d p_i \xi_i  \ge \epsilon
\end{equation}
for almost every  $\bx \in \Omega$ and every ${\bf \xi} \in \bR^d$ for which $|{\bf \xi}|=1$. Suppose further that $\omega \subset \Omega, T>0$, and that $f \in H^1_{\rm loc}(\bR,\Omega)$ is causal and satisfies $f(\bx,t)=0$ if $\tau|\bx-\bx_0|+t_0-t \ge 0$ for every $(\bx_0,t_0) \in \omega \times [0,T]$. Then the causal solution $u \in C^1(\bR,L^2(\Omega)^k \cap C^0(\bR,V)$ of \eqref{symmhyp} with coefficients $a,b,q$ and right-hand side $f$ vanishes in $\omega \times [0,T]$.
\end{theorem}

\noindent {\bf Example.} Linear acoustics provides an important example of the framework just described. Acoustic wave propagation does not include the memory effect modeled by the convolutional term in \eqref{symmhyp}, but illustrates several other features of the class of problems studied in this paper.

For this example, $d=3$, and we suppose that the domain $\Omega \subset \bR^3$ is smoothly bounded (the smoothness requirement can be relaxed to some extent, as discussed briefly in Appendix A). The momentum balance and constitutive laws of linear acoustics relate the excess pressure $p(t,\bx)$ and velocity fluctuations $\bv(t,\bx)=(v_1(\bx,t),v_2(\bx,t),v_3(\bx,t))^T$, $\bx \in \bR^3$, to mass density $\rho(\bx)$, bulk modulus $\kappa(\bx)$, constitutive law defect $g(t,\bx)$, and body force density $\bff(t,\bx)$ by 
\begin{eqnarray}\label{acoustics}
\rho\frac{\partial \bv}{\partial t} & = & -\nabla p + \bff, \nonumber\\
\frac{1}{\kappa}\frac{\partial p}{\partial t} & = & -\nabla \cdot \bv +g.
\end{eqnarray}
In this example, $k=4$, and 
\[
Au = {\rm diag}\left(\frac{1}{\kappa},\rho,\rho,\rho\right)u, \,\, u = 
\left(
\begin{array}{c}
p\\
v_1\\
v_2\\
v_3
\end{array}
\right) \in L^2(\Omega)^4
\]
defines a bounded self-adjoint positive-definite operator $A \in {\cal B}_{\rm symm}^+(L^2(\Omega)^4)$, provided that $\log \rho, \log \kappa \in L^{\infty}(\Omega)$.

The right-hand side of \eqref{acoustics} involves the matrix partial differential operator
\begin{equation}
\label{grad-div}
p(\nabla) = -\left(
\begin{array}{cccc}
0 & \frac{\partial}{\partial x_1} & \frac{\partial}{\partial x_2} & \frac{\partial}{\partial x_3} \\
\frac{\partial}{\partial x_1} & 0 & 0 & 0 \\
\frac{\partial}{\partial x_2} & 0 & 0 & 0 \\
\frac{\partial}{\partial x_3} & 0 & 0 & 0
\end{array}
\right).
\end{equation}
This matrix operator defines a skew-adjoint operator with domain $V = H^1_0(\Omega) \times H^1_{\rm div}(\Omega) \subset L^2(\Omega)^4$. [$H^1_{\rm div}(\Omega)$ consists of $\bv \in L^2(\Omega)^3$ for which the distribution $\nabla \cdot \bv \in L^2(\Omega)^3$. For the reader's convenience, we present a proof that $p(\nabla): V \rightarrow  L^2(\Omega)^4$ is skew-adjoint, in Appendix A.]

Finally, the source vector $f$ is defined by
\[
f(t)=(g(t,\cdot),\bff(t,\cdot))^T \in L^2(\Omega)^4,
\]
The acoustics system \eqref{acoustics} is symmetric hyperbolic and satisfies the conditions of Theorem \ref{symmhyp:wellposed}. 

Note that if $\kappa$ and/or $\rho$ are discontinuous, then the form of the equations \eqref{acoustics} immediately implies that no solutions of class $C^1$ may exist, even if the right-hand side $f$ (that is, the body force density $\bff$) is smooth. Physically reasonable fluid configurations thus exist for which solutions in the classical sense cannot be defined, for example piecewise homogeneous mixtures with jump discontinuities of density and/or bulk modulus across smooth interfaces. However, Theorem \ref{symmhyp:wellposed} provides a solution with physical sense, as the strain energy
\[
E(t) = \frac{1}{2}\int_{\Omega} \left[\rho \bv \cdot \bv + \frac{p^2}{\kappa}\right](\cdot,t) = \frac{1}{2}\langle u(t), A u(t)\rangle_{L^2(\Omega)^4}
\]
is well-defined for any such solution.

Theorem \ref{symmhyp:wellposed} gives  $u \in C^0(\bR,V)$, provided that $f \in H^1_{\rm loc}(\bR,L^2(\Omega)^4)$. To obtain differentiable dependence on the coefficients it is necessary to sacrifice one more time derivative. This sacrifice seems minor, as modeled energy sources (represented by the right-hand side in \eqref{symmhyp}) are in many cases quite smooth in time, reflecting the finite bandwidth of energy generation and recording equipment and the loss of high frequencies to dissipative mechanisms during propagation. Membership in $V = H^1_0(\Omega) \times H^1_{\rm div}(\Omega)$ implies that some (but not all!) traces of $u$ are well-defined. In the notation of Theorem \ref{symmhyp:tracereg}, given a smoothly embedded hypersurface $\Gamma$, choose $l=1$ and $m=m_D$ or $m=m_N$, where
\[
m_D(\bx) = [1, 0, 0, 0],\,\, m_N(\bx)=[0,n_1(\bx),n_2(\bx),n_3(\bx)], \,\,\bx \in \Gamma,
\]
and ${\bf n}=(n_1,n_2,n_3)^T$ is a smooth normal field on $\Gamma$. Both $S[m_D]$ and $S[m_N]$ extend continuously to $V$. Note that traces of tangential components of $\bv$ do not so extend.
The end result is that both ${\cal F}_{f,m_D}$ (pressure data) or ${\cal F}_{f,m_N}$ (normal velocity data) are well-defined and differentiable as functions of the coefficients, of class $C^{s-1}$ if $f \in H^s_{\rm loc}(\bR,L^2(\Omega)^4)$ (Theorem \ref{symmhyp:tracereg}). The domain of either version of ${\cal F}$ in this case is the open set $M_{\rm ac} \subset L^{\infty}(\Omega)^2$, defined by nondimensionalizing choices of scale $s_{\kappa}, s_{\rho} \in \bR_+$ and 
\[
M_{\rm ac}(C_*, C^*) = \{(\kappa,\rho) \in L^{\infty}(\Omega)^2: (s_{\kappa}\kappa)^{-1}, (s_{\rho}\rho)^{-1} < (C_*)^{-1}, s_{\kappa} \kappa, s_{\rho} \rho < C^* \mbox{ a.e. in }\Omega\};
\]
\[
M_{\rm ac} = \bigcup \{M_{\rm ac}(C_*, C^*) : 0 < C_* \le C^*\}.
\]
The content of Theorem \ref{symmhyp:speed} in this case is just what one would expect: waves move with speed at most
\[
\|\sqrt{\kappa/\rho}\|_{L^{\infty}(\Omega)}.
\]

We end this section with a brief overview of prior work. The well-posedness of symmetric hyperbolic systems with regular (smooth) coefficients has been well-understood for decades \cite[]{CourHil:62,Lax:PDENotes}. Some effort has also been devoted to describing solutions for first- or second-order hyperbolic systems with less-than-smooth (but still continuous) coefficients, mostly focused on the propagation of regularity along bicharacteristics, or the existence of traces (for example, \cite[]{BealsReed:82,Symes:83,BealsReed:84,Lasi:86,Lasi:87,BaoSy:91b,BaoSy:93a,HSmith:98}. 
Those works treating the continuity or differentiability of the solution as a function of the coefficients have mostly required more smoothness of the coefficients than is allowed in the present work, or dealt only with one-dimensional problems \cite[]{BamChavLai:79,Symes:86b,lew91,BaoSy:95,Salo:06,StefUhl:09}. An exception is \cite[]{FerSanShe:93}, in which (results which imply) a special case of Theorem \ref{symmhyp:tracereg} is established. A large number of articles have appeared on abstract first- or second-order hyperbolic equations and related inverse problems, for example \cite[]{LavRomShis:86,Choulli:91,LorenziRamm:01,RammKoshkin:01,Awawdeh:10,Orlovsky:10}, mostly using the theory of strongly continuous semigroups (or cosine operators, in the second order case) to obtain a hold on well-posedness. Of these, the closest in spirit to the present work is that of \cite{Orlovsky:10}, which treats second order systems via the cosine operator approach but observes that smoothness {\em in time} of data implies that weak solutions are strong solutions, analogous to a crucial intermediate result in our work. \cite{Choulli:91} established that a first-order abstract integro-differential initial value problem, resembling those of our abstract framework to be detailed in the next section, is well-posed, and also proved uniqueness for the solution of a related inverse problem. The form of the integral term is actually a generalization of our convolutional memory operator, and the requirement on the ``spatial'' operator is also a generalization (generates a strongly continuous semigroup - in the problem considered here, this operator is skew-adjoint, and generates a unitary group). However Choulli formulates an inverse problem which involves only determining lower-order terms, rather than the principal part, the inverse problem data does not involve time-like traces, and Choulli does not investigate the regularity of any analogue of ${\cal F}$.


The immediate predecessor of our work is Chapter 2 of the second-named author's PhD thesis \cite[]{stolk}, which treated second-order hyperbolic systems with  nonsmooth coefficients, including the displacement formulation of elastodynamics, beginning with the techniques of J.-L. Lions and his collaborators \cite[]{Lions:71,LionsMagenes:72} but going beyond well-posedness to study the dependence of solutions on coefficients. The results on strong convergence and G\^{a}teaux differentiability (Theorems \ref{strong} and \ref{smooth}) are direct translations of arguments from this source, as is the observation that convergence in measure of $L^{\infty}$ multipliers implies strong convergence of the corresponding multiplication operators on $L^2$ (Lemma \ref{mult}).

\section{A Class of Abstract First Order Evolution Equations and their Properties}

Let $H$ be a separable real Hilbert space, with inner product $\langle \cdot, \cdot \rangle$ and norm $\| \cdot \|$. We will denote by $V \subset H$ a dense subspace, itself a separable Hilbert space with inner product $\langle \cdot,\cdot \rangle_V$ and norm $\| \cdot \|_V$, defining a stronger topology on $V$ than that induced by $H$. ${\cal B}(H)$ is the Banach space of bounded linear operators on  $H$, with the operator (uniform) norm. ${\cal B}_{\rm symm}(H)$ is the subspace of bounded self-adjoint operators, and ${\cal B}_{\rm symm}^+(H)$ is the cone of bounded self-adjoint coercive operators. We suppose that 
\begin{itemize}
\item $A \in {\cal B}_{\rm symm}^+(H)$, $0< C_* \le C^*$ so that $C_*I \le A \le C^*I$.
\item $B \in {\cal B}(H)$;
\item $Q\in L^1(\bR, {\cal B}_{\rm symm}(H)) \cap C^0(\bR_+, {\cal B}(H))$, and is {\em causal}, that is, $Q(t) = 0$ for $t<0$;
\item $P$ is a skew-adjoint operator with domain $V$, for which the graph norm is equivalent to $\| \cdot \|_V$.
\end{itemize}

The kernel $Q$ defines a continuous linear operator $R: \lph \to \lph$ on the space
\[
\lph = \{v \in \llh: \mbox{ for each }T \in \bR, v \in L^2((-\infty,T],H) \}
\]
with the natural countably normed topology, by 
\begin{equation}\label{rdef}
R[u](t) = \int Q(t-s)u(s) \, ds.
\end{equation}
$R$ is well-defined since $Q$ is causal. Note that if $u \in \llh$ is {\em causal}, that is, ${\rm supp}\, u \subset [T,\infty)$ for some $T \in \bR$, then $u \in \lph$, $R[u]$ is causal also and ${\rm supp} \, R[u] \subset [T,\infty)$. The formal (distribution) adjoint 
\[
R^{\ast}[u](t)= \int Q(s-t)u(s)\,ds
\]
satisfies a similar condition: if  ${\rm supp}\, u \subset (-\infty,T]$ for some $T \in \bR$, then ${\rm supp} \, R^*[u] \subset (-\infty,T]$.

The components described above combine to yield the {\em formal evolution problem}: find an $H$-valued function of $t$, say $u$, which solves, in a suitable sense, 
\begin{eqnarray}
\label{formeq}
A u' + Pu + Bu + R[u] & = & f ,
\end{eqnarray}
in which the right-hand side $f$ is also an $H$-valued function of $t$.

We follow \cite[]{Lions:71,LionsMagenes:72,stolk} in defining {\em weak solutions} $u \in \llh$ of the formal evolution problem \eqref{formeq} by integration against smooth test functions. Inspection of \eqref{formeq} suggests that a similar constraint must also be placed on the right-hand side: henceforth, we assume $f \in \llh$. Because the operator kernel $Q$ may have unbounded support, we must constrain the growth of candidate members of $\llh$ on the negative half-axis. Accordingly, a {\em weak solution} of the formal evolution problem \eqref{formeq}
is a member of  $u \in \llh$ satisfying
\begin{itemize}
\item[1. ] $u \in \lph$;
\item[2. ]
\begin{equation}
\label{weakde}
\int  \langle u(t), (A \phi' + P \phi 
-B^\ast \phi - R^\ast [\phi])(t) \rangle \, dt = - \int \langle f(t), \phi(t) \rangle \, dt;
\end{equation}
for all $\phi \in C_0^\infty ( \bR, V)$.
\end{itemize}

Note that since $R^*[\phi]$ is supported in the half-axis $(-\infty,{\rm sup}\,{\rm supp}\,\phi]$, and is square-integrable, assumption 1. implies that the last term on the left-hand side of \eqref{weakde} is well-defined. 

We will reserve the term {\em strong solution} to designate a function $u \in C^1(\bR,H) \cap C^0(\bR,V) \cap \lph$ which solves \eqref{formeq} pointwise. Clearly a strong solution is a weak solution. Existence of a strong solution in this sense also implies that $f \in C^0(\bR,H)$.

Because of the causal assumption on the convolution kernel $Q$, existence of a causal weak solution vanishing for $t<T_0$ implies that the right-hand side $f$ must also be causal, in fact vanish for $t<T_0$. 

Concerning the fundamental questions of uniqueness and existence, we prove

\begin{theorem}\label{existence} Assume that $H$, $V$, $A$, $B$, $P$, and $Q$ have the properties listed above, and that $f \in \llh$ is causal. Choose $T_0$ so that ${\rm supp}\,f \subset [T_0,\infty)$. Then there exists a unique causal weak solution $u \in C^0(\bR, H)$ of (\ref{formeq}) with ${\rm supp} \,u \subset [T_0,\infty)$. For every $T  \ge T_0$, there exists $C_{T_0,T} \ge 0$ depending on $T_0$, $T$,  $C_*$, $C^*$,  $\|B\|_{{\cal B}(H)}$, and $\|Q\|_{L^1(\bR,{\cal B}(H))}$, so that for any $t \le T$,
\begin{equation}
\label{eqn:bashyp}
\|u(t)\|^2_H \le C_{T_0,T}\int_{-\infty}^{t} \,ds\,\|f(s)\|^2.
\end{equation}
\end{theorem}

Regularity in time of the right-hand side implies that the weak solution is strong:

\begin{theorem}\label{reghyp} In addition to the assumptions of Theorem \ref{existence}, suppose that  $f \in \lph \cap H^k_{loc}(\bR,H)$, $k \ge 1$. Then the causal weak solution $u$ of  \eqref{weakde} satisfies $u\in C^k(\bR,H) \cap C^{k-1}(\bR,V)$. For every $T  \ge T_0$, there exists $C_{T_0,T} \ge 0$ depending on $T_0$, $T$,  $C_*$, $C^*$,  $\|B\|_{{\cal B}(H)}$, and $\|Q\|_{L^1(\bR,{\cal B}(H))}$, so that for any $t \le T$,
\begin{equation}
\label{eqn:reghyp}
\|u^{(j)}(t)\|^2_V \le C_{T_0,T}\sum_{l=0}^{j+1}\int_{-\infty}^{t} \,ds\,\|f^{(l)}(s)\|^2, \,\,j=0,...k-1.
\end{equation}
In particular, $u \in C^1(\bR,H) \cap C^0(\bR,V)$ is a strong solution of \eqref{formeq}. For each $j=0,...,k$, $u^{(j)}$ is the weak solution of \eqref{formeq} with right-hand side $f^{(j)}$.
\end{theorem}

Regularity of weak solutions $u$ as functions of the coefficients $A,B,Q$ is expressed through estimates which are uniform over certain sets of problems (that is, certain sets of coefficients), which are important enough to merit a definition: given $0 < C_* \le C^*, 0 \le C_B, 0 \le C_Q$, define 
\[
{\cal P}(C_*,C^*,C_B,C_Q) = \{(A,B,Q) \in {\cal B}_{\rm symm}^{+}(H) \times {\cal B} \times [L^1(\bR, {\cal B}_{\rm symm}(H)) \cap C^0(\bR_+, {\cal B}(H))]: 
\]
\[
C_*I \le A \le C^*I, \|B\|_{{\cal B(H)}} \le C_B, \|Q\|_{L^1(\bR,H)}+\|Q\|_{L^{\infty}(\bR,H)} \le C_Q\};
\]
\begin{equation}
\label{admiss}
{\cal P} = \cup \{{\cal P}(C_*,C^*,C_B,C_Q): 0 < C_* \le C^*, 0 \le C_B, 0 \le C_Q\}.
\end{equation}

Strong convergence of the coefficient operators, in an appropriate sense, leads to uniform convergence of the corresponding weak solutions on compact sets in $\bR$. For simplicity, and because all of the examples we have in mind satisfy this restriction, we assume that all of the systems appearing the the formulation of this continuity result share the same ``spatial'' operator $P$. It is possible to weaken this assumption, that is, to approximate $P$ as well; we leave the formulation and proof of this (slightly) stronger result to the reader.

\begin{theorem}\label{strong}Suppose that $H,V$ are as described, $P:V\rightarrow H$ is skew-adjoint, $0 < C_* \le C^*, 0 \le C_B, 0 \le C_Q$, $\{(A_m,B_m,Q_m): m \in \bN\} \subset {\cal P}(C_*,C^*,C_B,C_Q)$, $(A,B,Q) \in {\cal P}(C_*,C^*,C_B,C_Q)$, and 
\begin{itemize}
\item[1. ] $ \lim_{m \to \infty} \| (A_m - A) w \| \to 0$ for all $w \in H$; 
\item[2. ] $\lim_{m \to \infty} \| (B_m - B) w \| \to 0$ for all $w \in H$;
\item[4. ] the convolution operators $R_m, R$ with kernels $Q_m,Q$ satisfy $\lim_{m \to \infty} \| R_m[w] - R[w]\|_{L^2((-\infty,T])} \to 0$ for any $T \in \bR$, all $w \in \lph$.
\end{itemize}
Let $u_m$, respectively $u$, be causal weak solutions of the differential equation \eqref{formeq} with coefficients $(A_m,P,B_m,Q_m)$, respectively $(A,P,B,Q)$. and (common) causal right-hand side $f \in \llh$.  Then for any choice of  $T_0 \le T \in \bR$, 
\[
\lim_{m \rightarrow \infty} \,\|u_m - u\|_{L^{\infty}([T_0,T],H)} = 0.
\]
\end{theorem}

With additional regularity of the right-hand side, solutions of \eqref{formeq} have directional (G\^{a}teaux) derivatives as functions of the coefficient operators:

\begin{theorem}\label{smooth} Suppose that $H,V$ are as described, $P:V\rightarrow H$ is skew-adjoint, $(A,B,Q) \in {\cal P}$, and $f \in H^1(\bR,H)$ is causal. Denote by $u \in C^1(\bR,H) \cap C^0(\bR,V)$ the causal strong solution of \eqref{formeq} with these choices of coefficients and right-hand side. Assume that $\delta A, \delta B \in {\cal B}(H)$, $\delta A$ is self-adjoint, and $\delta Q \in L^1(\bR,{\cal B}_{\rm symm}(H)) \bigcap C^0(\bR_+,{\cal B}(H))$, $\delta Q = 0$ for $t<0$. Define for $h \in \bR$
\[
A_h = A+h\delta A, \,\,B_h=B+h\delta B, \,\,Q_h=Q+h\delta Q.
\]
For sufficiently small $h$, $(A_,B_h,Q_h) \in {\cal P}$, so that the problem \eqref{formeq} with coefficients $A_h,P,B_h,Q_h$ and right-hand side $f$ has a unique (strong) solution $u_h \in C^1(\bR,H) \cap C^0(\bR,V)$. Denote by $\delta u \in C^0(\bR,H)$ the weak solution of the formal evolution problem
\begin{equation}\label{gateaux}
A \delta u' + P\delta u + B \delta u + R[\delta u] = -\delta A u' - \delta B  u - \delta R[u],
\end{equation}
in which $R$ ($\delta R$) is the convolution operator with kernel $Q$ ($\delta Q$), as usual. Then for any choice of  $T_0 \le T \in \bR$, 
\begin{equation}\label{newt}
\lim_{h \rightarrow 0}\,\left\|\frac{u_h-u}{h}-\delta u\right\|_{L^{\infty}([T_0,T],H)} = 0.
\end{equation}
\end{theorem}

With the correct choice of topology, the solution is (Fr\'{e}chet) differentiable as a function of the coefficients, assuming as in Theorem \ref{smooth} that the right-hand side is somewhat regular in $t$. To express this fact in a form most useful for application to inverse problems, introduce another Hilbert space $W$ (of ``measurements''), and a ``sampling'' map $S: V \rightarrow W$, assumed continuous, Of course $S$ induces a continuous map $:C^k(\bR,V) \rightarrow C^k(\bR,W)$ for any $k\in \bN$; we abuse notation by writing $S$ for this induced map also. 

For causal $f \in H^1_{\rm loc}(\bR,H)$, define ${\cal F}_f: {\cal P} \mapsto C^1(\bR,W)$ by
\begin{equation}
\label{forward}
{\cal F}_f[A,B,Q] = Su,
\end{equation}
in which $u \in C^1(\bR,H) \cap C^0(\bR,V)$ is the causal (strong) solution of \eqref{formeq} with coefficients $A,P,B,Q$ and right-hand side $f$. It follows from Theorem \ref{reghyp} that ${\cal F}_f$ is well-defined.

\begin{theorem}\label{regfwd} Suppose that $f \in H^k_{\rm loc}(\bR,H)$ is causal, $k \ge 2$, and view ${\cal P}$ as an open subset of the Banach space $ {\cal M} = {\cal B}_{\rm symm}(H) \times {\cal B}(H) \times [L^1(\bR, {\cal B}_{\rm symm}(H)) \cap C^0(\bR_+, {\cal B}(H))]$, with norm
\[
(A,B,Q) \mapsto \|(A,B,Q)\|_{\cal M} = \|A\|_{{\cal B}(H)} + \|B\|_{{\cal B}(H)} + \max ( \|Q\|_{L^1(\bR_+,{\cal B}(H))} , \|Q\|_{L^{\infty}(\bR_+,{\cal B}(H))} ). 
\]
Then for any $T_0 \le T_1 \in \bR$, ${\cal F}_f \in C^{k-1}({\cal P}, C^0([T_0,T_1],W))$.
\end{theorem}

\noindent {\bf Remark:} Note that this theorem expresses ``differentiability with a loss of one derivative'', in two senses. First, taking $k=2$, even though $f \in H^2_{\rm loc}(\bR,H)$, the map ${\cal F}$ is only once differentiable. Second, even though ${\cal F}$ {\em a priori} takes values in $C^1(\bR,W)$, it is of class $C^1$ only as a map with range $C^0(\bR,W)$. In these respects Theorem \ref{regfwd} is sharp - counterexamples can be adduced to show that neither of these losses-of-derivative are artifacts of the proof.

\section{The Energy Inequality}
Define the {\em energy} $E(t)$ of a weak solution $u$ of (\ref{formeq}) by 
\begin{equation}\label{edef}
E(t) = \half \langle u(t),A u(t) \rangle
\end{equation}
It follows from the definition of weak solution that $E$ is well-defined almost everywhere, and locally integrable. Because $A$ is positive-definite,
\begin{equation}
\label{enorm}
C_*\|u(t)\|^2 \le E(t) \le C^*\|u(t)\|^2
\end{equation}
hold for almost all $t\in \bR$, for suitable $C^* \ge C_* > 0$.

\noindent {\bf Remark.} In the linear acoustics example presented in the previous section, $E(t)$ is precisely the mechanical energy of the acoustic field at time $t$.

In this and the following sections, we will use $C$ to denote a generic nonnegative constant depending on $C_*,C^*,  \|B\|_{{\cal B}(H)}$, and $\|Q\|_{L^1(\bR,{\cal B}(H))}$, and possibly on other quantities as noted.

We shall repeatedly use following property of weak solutions:

\begin{proposition}\label{spacesmooth} Suppose that $u \in \llh$ is a weak solution of \eqref{formeq}. Then for any $\eta \in C_0^\infty(\bR)$, $\eta \ast u$ is a strong solution, and in fact
\begin{equation}\label{smprop}
\eta \ast u \in C^{\infty}(\bR,V).
\end{equation}
\end{proposition}
\noindent {\bf Remark:} In applications such as acoustics, or viscoelasticity to be discussed later, $H$ and $V$ are function spaces, and membership in $V$ entails additional regularity beyond that required for membership in $H$. In such applications, the content of this theorem is that smoothing in time also ``smooths in space'', in the sense that the values of the smoothed weak solution are confined to the subspace $V$. 

\begin{proof} Choose the test function $\phi$ in \eqref{weakde} to have the special form $\phi(s) = \eta(t-s)w$, where $t \in \bR$, $w \in V$, and $\eta \in C_0^{\infty}(\bR)$. Then
\begin{eqnarray}\label{foo-bar}
\langle \eta \ast u(t),Pw \rangle 
& = & \left\langle \int \eta(t-s)u(s)\, ds, Pw \right\rangle \nonumber \\
& = & \int \langle u(s), P(\eta(t-s)w)\rangle \,ds\nonumber \\
& = & \int \Big[ \langle u(s), A \eta'(t-s)w + B^{\ast} \eta(t-s)w \nonumber \\
& &+R^{\ast}[\eta(t-\cdot)w](s) \rangle \Big] \, ds -  \langle \eta \ast f (t),w\rangle \nonumber\\
& = & \langle A(\eta' \ast u)(t)+B(\eta \ast u)(t) + R[\eta \ast u](t),w\rangle,
\end{eqnarray} 
where the third equality is simply a rearrangement of \eqref{weakde} with the special choice of test function $\phi(s) = \eta(t-s)w$ mentioned above, and the last holds amongst other reasons because $R$ is also convolutional. 
The right-hand side of \eqref{foo-bar} is bounded by a $w$-independent multiple of $\|w\|_H$, therefore so is the left. Therefore $\eta \ast u$ takes values in the domain ${\cal D}(P^*)$ of the adjoint $P^*$ for any $\eta \in C_0^{\infty}(\bR)$. But $P$ is skew-adjoint, so ${\cal D}(P^*) = {\cal D}(P) = V$. Therefore we can shift $P$ to the left-hand side of the inner product on left-hand side of \eqref{foo-bar}, to obtain 
\begin{equation}
\label{foo-bar1}
\langle -P(\eta \ast u)(t),w \rangle = \langle A(\eta' \ast u)(t)+B(\eta \ast u)(t) + R[\eta \ast u](t) - (\eta \ast f)(t),w\rangle.
\end{equation}
As \eqref{foo-bar1} holds for every $w \in V$ and $V \subset H$ is dense, it follows that 
\begin{equation}
\label{foo-bar2}
-P(\eta \ast u)(t) = A(\eta' \ast u)(t)+B(\eta \ast u)(t) + R[\eta \ast u](t) - (\eta \ast f)(t),
\end{equation}
that is, \eqref{formeq} is satisfied pointwise in $t$. 

Since the right-hand side of \eqref{foo-bar2} is of class $C^{\infty}(\bR,H)$, so is the left-hand side. 
Denote by $T_{\delta t} \in {\cal B}(L^2_{\rm loc}(\bR,H))$ the translation operator by $\delta t$, defined by $T_{\delta t}u(t)=u(t+\delta t)$. For each $k \in \bN$, let 
\[
\Delta_{k,\delta t} = \delta t^{-k} \sum_{j=-m_k}^{m_k} a_{k,j} T_{\delta t}^j 
\] 
be a finite difference operator for which $u \in C^0(\bR,H)$ has $k$ derivatives if and only if $\Delta_{i,\delta t} u$ converges pointwise in $t$ for each $i=1,...,k$ to $u^{(i)}$ as $\delta t \rightarrow 0$. From \eqref{foo-bar2} it follows that for each $k \in \bN$,
\begin{eqnarray}
\label{foo-bar3}
-\Delta_{k,\delta t}(P(\eta \ast u))(t) &=& -P (\Delta_{k,\delta t}(\eta \ast u)(t)) \nonumber\\
&=& \Delta_{k,\delta t}(A(\eta' \ast u)+B(\eta \ast u) + R[\eta \ast u] - \eta \ast f)(t).
\end{eqnarray}
Since the right-hand side of \eqref{foo-bar3} converges for each $t$, so does the left-hand side, both versions, as $\delta t \rightarrow 0$. Thus $\Delta_{k,\delta t}(\eta \ast u)(t)$ converges in the graph norm of $P$, hence in the sense of $\|\cdot\|_V$, as $\delta t \rightarrow 0$ for each $t \in \bR, k \in \bN$. Thus $\eta \ast u \in C^{\infty}(\bR,V)$.
\end{proof}

\begin{proposition}\label{energy}Let $u \in \llh$ be a weak solution of \eqref{formeq}, $E \in L_{\rm loc}^1(\bR)$ its energy as defined in \eqref{edef}.  Then
\begin{itemize} 
\item after modification on a set of measure zero, $E$ is continuous;
\item if in addition $f$ is causal, $\mbox{supp }f \subset [T_0,\infty)$, then for any $T \in \bR$ there exists $C_{T_0,T}\ge 0$ so that for $t\in(-\infty,T]$,
\begin{equation}\label{energyest}
E(t) \leq C_{T_0,T}\int_{-\infty}^t \|f\|^2.
\end{equation}
\end{itemize}
\end{proposition}

\begin{proof}Let $\eta_n \in C_c^\infty(\bf{R})$ be an approximate identity, that is, $\eta_n(t) = n \eta(n t)$, where
\begin{equation}
\label{etahyp}
\eta \in C_0^\infty(\bR), \,\,
\eta \geq 0, \,\,
\int \eta (t) \, dt = 1, \,\, \mbox{supp }\eta \subset [-1, 1].
\end{equation}
Define
\begin{equation}\label{energysequence}
E_n(t) = \half \langle (\ee \ast u)(t), A (\ee \ast u)(t) \rangle.
\end{equation}
Since $\ee \ast u \rightarrow u$ in $\llh$, $E_n \rightarrow E$ in $L^1_{\rm loc}(\bR)$. Thanks to Proposition  \ref{spacesmooth}, $\ee \ast u \in C^{\infty}(\bR,V)$ is a strong solution of \eqref{formeq} for each $n\in \bN$.

For each $n$, $E_n$ is smooth; differentiating $E_n$, obtain for any $s, t \in \bR$
\begin{eqnarray}
\label{endiff}
E_n(t)-E_n(s)& = & \int_s^t \frac{d E_n}{ds}(\tau)\,d\tau \nonumber\\
& = & \int_s^t  \langle (\ee \ast u)(\tau), A (\ee \ast u)'(\tau) \rangle \,d\tau\nonumber \\
& = & \int_s^t \left[- \langle  (\ee \ast u)(\tau), (P(\ee \ast u)(\tau) + B(\ee \ast u)(\tau)+R[\ee \ast u]) (\tau) \rangle \right.\nonumber\\
& & + \left. \langle (\ee \ast f)(\tau), (\ee \ast u)(\tau) \rangle\right] \,d\tau \nonumber \\
& = & \int_s^t \left[- \langle  (\ee \ast u)(\tau), (B(\ee \ast u)(\tau)+R[\ee \ast u]) (\tau)\rangle \right.\nonumber\\
& & + \left. \langle (\ee \ast f)(\tau), (\ee \ast u)(\tau) \rangle\right] \,d\tau. 
\end{eqnarray}
thanks to the skew-symmetry of $P$.

Since convolution with $\eta$ commutes with the convolution operator $R$, and with the actions of the other operators appearing in (\ref{formeq}), the identity \eqref{endiff} implies that
\begin{eqnarray}
| E_n(t) - E_n(s) | & = & \left| \int_s^t \frac{d E_n}{d\tau}(\tau) \, d\tau \right| \nonumber \\
& \leq & \int_s^t \bigg[ | \langle \ee \ast (Bu(\cdot))(\tau), \ee \ast u(\tau) \rangle |  \nonumber \\
& & + | \langle (\ee \ast R[u]) (\tau), \ee \ast u (\tau) \rangle | \nonumber \\
& & + | \langle \ee \ast f(\tau), \ee \ast u(\tau) \rangle | \bigg] \, d\tau \nonumber \\
\label{ediff}
& \leq & (\|B \|+1) \left(\int_s^t \| (\ee \ast u)(\tau) \|^2 \, d\tau \right) \nonumber\\
& & + \int_s^t (\| (\ee \ast R[u]) (\tau)\|^2 + \| (\ee \ast f)(\tau) \|^2) \, d\tau
\end{eqnarray}
Since $u$, $R[u]$, and $f$ are locally square-integrable, for each $t \in \bR$ and $\epsilon>0$, there exist $\Delta t(t,\epsilon)>0$ and an $N(t,\epsilon) \in \bf{N}$ so that for  $|s-t| < \Delta t(t,\epsilon)$ and $n > N(t,\epsilon)$, 
\[
\int_{s-1/n}^{t+1/n}  \| u\|^2 < \epsilon,\,\,\int_{s-1/n}^{t+1/n}  \| R[u]\|^2 < \epsilon,\,\,\mbox{and}\,\,\int_{s-1/n}^{t+1/n}  \| f\|^2 < \epsilon,
\]
whence (\ref{ediff}) implies that for $n>N(t,\epsilon)$, $|s-t|<\Delta t(t,\epsilon)$,
\begin{equation}
\label{equi}
|E_n(t)-E_n(s)| < C\epsilon.
\end{equation}
Continuity of $E_n$ implies existence of $\overline{\Delta t}(t,\epsilon)>0$ so that for $|s-t|<\overline{\Delta t}(t,\epsilon)$ and $n \le N$, $|E_n(t)-E_n(s)| < C\epsilon$. Thus the inequality (\ref{equi}) holds for all $n\in \bf{N}$ if $s$ satisfies $|s-t| < \min(\Delta t(t,\epsilon), \overline{\Delta t}(t,\epsilon))$. 
Since $t \in \bR, \epsilon>0$ are arbitrary, we have shown that the sequence $\{E_n\} \subset C^0(\bR)$ is equicontinuous.

Choose $T_0 \le T \in \bR$: it follows from the definition \eqref{energysequence} and Young's inequality that 
\begin{equation}
\label{indep}
\int_{T_0}^{T} E_n \le C \int_{T_0-1}^{T+1} \|u\|^2
\end{equation}
is bounded independently of $n$. For $t \in [T_0,T]$,
\begin{equation}
\label{stupid}
(T-T_0)E_n(t) = \int_{T_0}^T (E_n(t)-E_n(s))\, ds + \int_{T_0}^TE_n(s)\, ds
\end{equation}
Apply Young's inequality to the convolutions with $\ee$ appearing in the inequality (\ref{ediff}) to conclude that for $T_0 \le s,t \le T$, 
\begin{equation}
\label{ediffbig}
|E_n(t)-E_n(s)| \le C \int_{T_0-1}^{T+1} (\|u\|^2 + \|R[u]\|^2 + \|f\|^2).
\end{equation}
Taken together, (\ref{indep}), (\ref{stupid}) and (\ref{ediffbig})  imply that $\{E_n\}$ is a bounded subset of $C^0([T_0,T])$. According to Ascoli's theorem, $\{E_n\}$ is precompact in $C^0([T_0,T])$, hence a subsequence converges uniformly to a continuous limit. Since the subsequence is necessarily also $L^1$-convergent, and $T\in \bR$ is arbitrary, the first assertion of the theorem is established.

In view of the continuity of $E$, we may take the limit $n \rightarrow \infty$ on both sides of the inequality (\ref{ediff}) along the uniformly convergent subsequence whose existence we have just established. Since $\ee \ast u \rightarrow u$ in $L^2([T_0,T],H)$, the right hand side converges, and we obtain
\begin{equation}
\label{ediffprec}
|E(t)-E(s)| \le C \int_{s}^{t} (\|u\|^2 + \|R[u]\|^2 + \|f\|^2).
\end{equation}
We have assumed $u$ to be causal, but this assumption has not appeared in the reasoning up to now. It allows us to take $s \rightarrow -\infty$ in \eqref{ediffprec}. In view of the equivalence of $\sqrt{E}$ and the norm $\| \cdot \|$ (inequalities (\ref{enorm})), the inequality (\ref{ediffprec}) implies that
\[
E(t) \le C \int_{-\infty}^t (E + \|f\|^2).
\]

Gronwall's inequality then yields the second conclusion.
\end{proof}

\begin{corollary} The energy $E$ of a weak solution $u$ of \eqref{formeq}, as defined above, satisfies for any $s,t \in \bR$
\begin{equation}\label{eident}
E(t) - E(s) = \int_s^t \langle -Bu(\tau)-R[u](\tau)+f(\tau),u(\tau)\rangle \,d\tau.
\end{equation}
\end{corollary}

\begin{proof} Continuity of $E$ and convergence of $\ee \ast u$ to $u$ in $\llh$ allows us to take limits on both sides of \eqref{endiff}.
\end{proof} 

\begin{corollary}\label{unique1} Suppose that $u_1,u_2 \in \llh$ are {\em causal} weak solutions of (\ref{formeq}). Then $u_1=u_2$.  
\end{corollary}
\begin{proof} The conclusion follows immediately from the energy inequality (\ref{energyest}), applied to the difference $u=u_1-u_2$, which is a weak solution with $f\equiv 0$. 
\end{proof}

\begin{corollary}\label{cont} Suppose that $u \in \llh$ is a causal weak solution of (\ref{formeq}). Then $u \in C^0(\bR,H).$\end{corollary}

\begin{proof} For $\delta t \in \bR$, denote by $u_{\delta t}$ the member of $\llh$ defined by $u_{\delta t}(t)=u(t+\delta t)$. Then $u_{\delta t}$ is a causal weak solution (the only one, thanks to Corollary \ref{unique1}) of (\ref{formeq}) with $f$ replaced by $f_{\delta t} \in \lh$, defined by $f_{\delta t}(t)=f(t+\delta t)$. The translation group acts strongly continuously on $L^2$, i.e. $\|f_{\delta t}-f\|_{\lh} \rightarrow 0$ as $\delta t \rightarrow 0$. Since the difference $u_{\delta t}-u$ is a causal solution of (\ref{formeq}) with right-hand side $f_{\delta t}-f$, it follows immediately from (\ref{energyest}) that $\|u_{\delta t}(t)-u(t)\|_H \rightarrow 0$ as $\delta t \rightarrow 0$ for any $t \in \bR$, that is, $u \in C^0(\bR)$.
\end{proof}

\begin{corollary}\label{equicont} Suppose that 
\begin{itemize}
\item[1. ] ${\cal K} \subset {\cal B}(H)$ is a bounded set;
\item[2. ] ${\cal L} \subset {\cal B}(V,H)$ is a bounded set of skew-adjoint operators on $H$ with (common) domain $V$, whose graph norms are all equivalent (to each other and to the norm in $V$);
\item[3. ] ${\cal M} \subset {\cal B}(H)$ is a bounded set of self-adjoint, uniformly positive definite operators: there exist constants $0 < C_* \le C^*$ so that for all $A \in {\cal M}$, 
\[
C_*I \le A \le C^*I;
\] 
\item[4. ] ${\cal Q} \subset L^1(\bR,{\cal B}_{\rm symm}(H)) \bigcap C^0(\bR_+,{\cal B}(H))$ is a bounded set of causal operator-valued functions: if $Q \in {\cal Q}$, then $Q(t)=0$ for $t<0$.
\end{itemize}
Let the set ${\cal P} \subset {\cal M} \times {\cal L} \times {\cal K} \times {\cal Q}$ parametrize a family of formal evolution problems of for \eqref{formeq}, with coefficients $A \in {\cal M}, P \in {\cal L}, B \in {\cal K},$ and $Q \in {\cal Q}$, with common right-hand side $f \in \lh$, and let ${\cal U} \subset \llh$ be a corresponding family of causal weak solutions. Then ${\cal U} \subset C^0(\bR,H)$ is equicontinuous.
\end{corollary}
\begin{proof} That ${\cal U} \subset C^0(\bR,H)$ is the content of the last Corollary. It follows from the proof of the basic energy estimate \eqref{energyest} that the constant $C$ appearing in its right-hand side may be chosen uniform over ${\cal P}$ - indeed, the bounds defining the sets listed in conditions 1-4 above are precisely those on which our constants, canonically notated $C$, depend. Therefore \eqref{energyest} implies that for $u \in {\cal U}$,
\[
\|u(t+\delta t)-u(t)\|^2 \le \frac{1}{C_*}E_{u_{\delta t}-u}(t) \le C \int_{-\infty}^t \|f_{\delta t}-f\|^2 = C\int_t^{t+\delta t}\|f\|^2
\]
from which a uniform modulus of continuity follows. 
\end{proof}

Additional regularity in time of the right-hand side $f$ translates into additional regularity of the solution.

\begin{corollary}\label{ck} Suppose that $f \in L^2_{T_0}(\bR,H) \cap H^k_{\rm loc}(\bR,H)$, $k \in \bN$, and $u\in \llh$ is a weak solution of (\ref{formeq}). Then $u \in C^k(\bR,H)$; moreover, for $j=1,...k$, $u^{(j)}$ is the weak solution of \eqref{weakde} with right-hand side $f^{(j)}$.\end{corollary}

\begin{proof} Corollary \ref{cont} is the case $k=0$. 

The case $k=1$ provides the induction step, so we discuss it first. Since $f \in H^1_{\rm loc}(\bR,H)$, using the notation of the proof of Corollary \ref{cont},
\[
\frac{1}{\delta t}(f_{\delta t} - f) \rightarrow f' \mbox{ as } \delta \rightarrow 0
\]
in mean square. The energy estimate \eqref{energyest} implies that for each $t$,
\[
\frac{1}{\delta t}(u_{\delta t}(t)-u(t)) 
\]
has a limit as $\delta t \rightarrow 0$, whence $u$ is differentiable and $u'$ satisfies \eqref{weakde} with right-hand-side $f'$. The previous Corollary shows that $u'$ is continuous.

This argument also serves as the induction step to establish the assertion of the Corollary for $k>1$. \end{proof}

Smoothing the solution in time leads to a different sort of ``regularity'':

\begin{corollary}\label{pseudohyp} Suppose that $T_0 \in \bR, f \in L^2_{T_0}(\bR,H)$, $\Delta T > 0$, $\eta \in C^{\infty}_0((-\Delta T, \Delta T))$, and let $u$ denote a causal weak solution of  \eqref{weakde}. Then for every $T  \ge T_0$, there exists $C_{T_0,T,\Delta T} \ge 0$ depending on $T_0$,$T$, $\Delta T$, $\| \eta\|_{L^1(\bR)}$, $\|\eta'\|_{L^1(\bR)}$, $C_*$, $C^*$,  $\|B\|_{{\cal B}(H)}$, and $\|Q\|_{L^1(\bR,{\cal B}(H))}$, so that for any $t \le T$,
\begin{equation}
\label{eqn:pseudohyp}
\|(\eta \ast u)(t)\|^2_V \le C_{T_0,T,\Delta T}\int_{-\infty}^{T+\Delta T} \,ds\,\|f(s)\|^2.
\end{equation}
\end{corollary}

\noindent {\bf Remark:} Note that this inequality is a pointwise bound on $\eta \ast u$ in the norm of $V$. In applications, $V$ is compactly embedded in $H$, so this is potentially a much stronger statement than the obvious $H$-norm bound which follows directly from \eqref{energyest}.

\begin{proof}
The left hand side makes sense, of course, thanks to Proposition \ref{spacesmooth}. The identity\eqref{foo-bar} may be re-written as
\[
-\langle P(\eta \ast u)(t),w \rangle 
 =  \int \Big[ \langle u(s), -A \eta'(t-s)w + B^{\ast} \eta(t-s)w 
+R^{\ast}[\eta(t-\cdot)w](s) \rangle \Big] \, ds -  \langle \eta \ast f (t),w\rangle
\]
for any $w \in H$ (not just $V$!). Since ${\rm supp}(\eta) \subset [-\Delta T,\Delta T]$, for any $t \le T$,,
\begin{eqnarray}\label{foo-barmk2}
|\langle P(\eta \ast u)(t),w \rangle | & \le &C \|w\| \int_{-\infty}^{T+\Delta T}\,ds\, (\|u(s)\|+\|f(s)\|)\\
& \le &  C_{T_0,T,\Delta T} \|w\|\left(\int_{-\infty}^{T+\Delta T}\,ds\,\|f(s)\|^2 \right)^{\frac{1}{2}}
\end{eqnarray} 
thanks to the the energy estimate \eqref{energyest} and by-now familiar use of Young's inequality and the assumptions on the various operators and quantities in the problem formulation. Now choose $w = P(\eta \ast u)(t)$ to obtain a bound on its $H$-norm. In combination with \eqref{energyest} and Young's inequality, this estimate implies a bound of the required form on $(\eta \ast u)(t)$ in the graph norm of $P$, that is, in the norm of $V$.
\end{proof}

\section{Existence of Weak Solutions: Proofs of Theorems \ref{existence} and \ref{reghyp}}

The proof of existence follows the pattern laid out by \cite{Lions:71}, which in turn echos Cauchy's proof of the fundamental theorem of ordinary differential equations. We define a Galerkin method, show that it converges, and finally that the limit is a weak solution. Note that no rate of convergence follows from this argument; in fact it is easy to see that none can be expected. Of course, Proposition \ref{energy} has already assured that the solution so constructed is the only solution.

\begin{proof} of Theorem \ref{existence}:

In view of the energy estimate (Proposition \ref{energy}), at most one such solution exists, and any  sequence of weak solutions, corresponding to an $\lh$-convergent sequence of right hand sides $f$, must itself be $\lh$-convergent. Therefore it suffices to establish existence of solutions for a $\lh$-dense set of right hand sides. In particular, we may assume that $f \in C^0(\bR,H)$, without loss of generality.

Since $V$ is separable (with respect to the graph norm of $P$) and $V \subset H$ is dense, countable linearly independent subsets $\{ w_k \}_{k=1}^\infty \subset V$ exist for which finite linear combinations are dense in $V$, hence in $H$. Without loss of generality, assume that $\{ w_k \}_{k=1}^\infty$ is ($H$-) orthonormal: $\langle w_k,w_l\rangle=\delta_{kl},\,k,l \in {\bf N}$. 

Define  $m \times m$ matrices $A^m $ (symmetric positive definite), $P^m$ and $B^m$  by
\begin{eqnarray}
A^m_{kl} & = &\langle Aw_k,w_l \rangle,\\
P^m_{kl} & = &\langle Pw_k,w_l \rangle,\\
B^m_{kl} & = &\langle Bw_k,w_l \rangle,
\end{eqnarray}
for $1 \le k,l \le m$, and the operator $R^m$ on $L_{{\rm loc}}^2(\bR)^m$ defined analogously to (\ref{rdef}) by 
\[
R^m U^m(t) = \int_{-\infty}^t \langle Q^m(t-s) U^m(s) \,ds,\,\,Q^m_{kl}(t) \rangle \, ds = \langle Q(t)w_k,w_l\rangle,\,\,1 \le k,l \le m.
\]
Note that $Q^m \in L^1(\bR, {\cal B}_{\rm symm}(\bR^m)) \bigcap C^0(\bR_+,{\cal B}(\bR^m))$ is causal ($Q^m(t)=0, t<0$).

Define $F^m \in C^0(\bR)^m$ by 
\[
F^m_k(t)=\langle f(t),w_k\rangle,\,\,1 \le k \le m.
\]

A minor modification of a standard contraction mapping argument (see for example \cite{coddington-levinson}) shows that for each $m \in {\bf N}$, the initial value problem 
\begin{eqnarray}\label{fem}
A^m \frac{d U^m}{dt} + P^m U^m + B^mU^m + R^mU^m & = & F^m, \nonumber \\
U^m(t) & = & 0,\,\, t < T_0.
\end{eqnarray}
has a unique solution $U^m \in C^1(\bR,\bR^m)$.

For each $m \in {\bf N}$, define $u_m \in C^1(\bR,V), f_m \in C^0(\bR,H)$ by
\begin{equation*}
u_m(t) = \sum_{k=1}^m U^m_k(t) w_k, \,\, f_m(t)= \sum_{k=1}^m F^m_k(t)w_k.
\end{equation*}
Then the system \eqref{fem} satisfied by $U^m$, together with the $H$-orthonormality of $\{w_k\}$, implies that $u_m$ is the weak solution of the evolution equation (\ref{formeq}) with right-hand side $f_m$. The energy estimate (\ref{energyest}) shows that the sequence $u_m$ is bounded in $\llh$, hence by the Tychonoff-Alaoglu theorem and a diagonal process argument  weakly precompact in $\llh$. Denote by $u_{m(l)}$ a weakly convergent subsequence, and by $u$ its weak limit. Since $u_m(t) =0$ for $t<T_0$ and all $m \in {\bf N}$, the same is true for $u$.

To see that the limit $u$ is a weak solution of (\ref{formeq}), introduce for each $m_0 \in {\bf N}$ test functions $\psi$ of the form
\begin{equation}\label{psidef}
\psi = \sum_{k=1}^{m_0} \phi_k \otimes w_k, \hspace{.15in} \phi_k \in C_0^\infty (\bR).
\end{equation}
For $l$ sufficiently large that $m(l)>m_0$, $\langle f^m(t),\psi(t) \rangle = \langle f(t),\psi(t)\rangle$, $\langle P \uml, \psi \rangle = - \langle \uml, P\psi \rangle$, etc. So
\[
\int \langle u_{m(l)}, A \psi' + P \psi  - B^\ast \psi - R^{\ast}[\psi] \rangle \, dt =  - \int \langle f, \psi \rangle \, dt.
\]
Letting $l \to \infty$, it follows that $u$ satisfies \eqref{weakde} for all test functions $\psi$ of the form given in equation \eqref{psidef}. Since linear combinations of $w_m$'s are dense in $V$, the set of functions of the form \eqref{psidef} is dense in $C_0^1(\bR,V)$, whence $u$ is a weak solution of \eqref{formeq}.
\end{proof}

\begin{proof} of Theorem \ref{reghyp}:

We give the proof for $k=1$; the general case follows by a straightforward induction argument. 

According to 
Choose a Dirac sequence $\{\eta_n: n \in \bN\} \subset C^{\infty}_0(\bR)$ as in the proof of Proposition \ref{energy}. According to Corollary \ref{pseudohyp},  $\eta_n \ast u \in C^k(\bR,V)$  for each $n \in \bN$. According to Corollary \ref{ck}, $u \in C^1(\bR,H)$,
so the first term on the right-hand side of \eqref{foo-bar} may be integrated by parts to yield
\begin{eqnarray}
\label{foo-barmk3}
\langle (\eta_n \ast u),Pw \rangle &  = & \eta \ast \langle u, Pw \rangle\\
& = & \eta_n \ast \langle A u'  +B  u +R[u] -f, w\rangle
\end{eqnarray}
for any $w \in V$. Both sides are continuous in $t$, so the limit as $n\rightarrow \infty$ is  valid pointwise. Thus
\begin{eqnarray}
\label{wkstr}
\langle u(t),Pw \rangle &  = & \langle A u'(t)  +B  u(t) +R[u](t) -f(t), w\rangle
\end{eqnarray}
for any $w \in V$, $t\in \bR$. Pointwise bounds on $\|u(t)\|$ (Proposition \ref{energy}), $\|u'(t)\|$ (Corollary \ref{ck}) and the standing assumptions on the various operators imply that
\[
\langle u(t),Pw \rangle \le C_{T_0,T} \|w\| \left( \int_{-\infty}^T [\|f\|^2+\|f'\|]^2\right)^{\frac{1}{2}} 
\]
for all $w \in V$, $t\in (-\infty,T]$ (and any $T$, but the constant $C_{T_0,T}$ depends on $T_0$ and $T$), which shows that $u(t) \in {\cal D}(P^*) = {\cal D}(P) = V$ for each $t \in \bR$. Thus
\[
\langle Pu(t),w \rangle \le C_{T_0,T} \|w\| \left( \int_{-\infty}^T [\|f\|^2+\|f'\|]^2\right)^{\frac{1}{2}} 
\]
for all $w \in H$, $t \in (-\infty,T]$; taking $w = Pu(t)$, obtain the pointwise bound \eqref{eqn:reghyp} for $k=1$. That $u \in C^0(\bR,V)$ follows exactly as in the proof of Corollary \ref{cont}, via the strong continuity of the translation group on $H^1(\bR,H)$.
\end{proof}

\section{Continuous Dependence on Parameters: proofs of Theorems \ref{strong}, \ref{smooth}, and \ref{regfwd}}

The proof of Theorem \ref{strong} divides into three steps. In the first step (Lemma \ref{weak.conv}), we show that the assumptions of the theorem imply $L_{\rm loc}^2(\bR,H)$-weak convergence of the sequence of solutions. The second step leverages this result to show that the sequence converges $H$-weakly, pointwise (Lemma \ref{pw.weak.conv}). The last step combines these two observation with strong convergence of the coefficients to obtain convergence in norm of the solution sequence. 

\begin{lemma}\label{weak.conv}
Under the conditions of Theroem \ref{strong}, $u_m$ converges weakly to $u$ in $L_{\rm loc}^2(\bR,H)$.
\end{lemma}

\noindent {\bf Remark.} In order that $R_m$ converge to $R$ pointwise, as assumed in the statement of the preceding theorem, it is sufficient that $Q_m \rightarrow Q$ uniformly in $\bR_+$.

\begin{proof}


The bounds implied by the energy estimate (Proposition \ref{energy}) are uniform over bounded sets of coefficients as described in the statement of the theorem. Therefore $\{u_m\}$ is bounded in $\llh$, hence has an $\llh$-weakly convergent subsequence $\uml$, with limit $\bar{u} \in \llh$. Note that $\llh$-weak convergence implies convergence in the sense of $H$-valued distributions on $\bR$. Choose a test function $\phi \in C_0^{\infty}(\bR,V)$: then 
\begin{eqnarray}\label{foo-foo}
-\int \langle f(s),\phi(s) \rangle \, ds 
& = & \int  \langle \uml(s), (A_{m(l)} \phi' + P_{m(l)} \phi 
-B_{m(l)}^\ast \phi - R_{m(l)}^\ast [\phi])(s) \rangle \, ds \\
& = &\int \langle \bar{u}(s), (A \phi' + P \phi 
-B^\ast \phi - R^\ast [\phi])(s) \rangle \, ds \\
& & + \int \langle (\uml(s)-\bar{u}(s)), (A \phi' + P \phi 
-B^\ast \phi - R^\ast [\phi])(s) \rangle \, ds \\
& & + \int \Big\langle \uml(s), ((A_{m(l)}-A) \phi' + (P_{m(l)}-P) \phi \\
& &-(B_{m(l)}-B)^\ast \phi - (R_{m(l)}-R)^\ast [\phi])(s) \Big\rangle \, ds
\end{eqnarray}
The second term  vanishes in the limit $l \rightarrow \infty$ because of the weak convergence of $\uml$ to $\bar{u}$.  
The coefficients $A_m, ...$ range over bounded sets of operators, so we may replace $\phi$ and $\phi'$ in the third term with simple $V$-valued functions, taking finitely many values, at the price of an arbitrarily small perturbation in this term, uniformly in $l$. However the strong convergence of the coefficient operators assumed in the statement of the theorem then implies that the resulting integrals become arbitrarily small as $l \rightarrow \infty$. Thus  $\bar{u}$ is a weak solution of the problem \eqref{formeq}, and must therefore be the same as the (unique) weak solution $u$ constructed in the preceding section. Thus no other weak accumulation point of the bounded sequence $\{u_m\}$ may exist, hence $u_m \rightharpoonup u$ in $\llh$ as claimed.
\end{proof}

\begin{lemma}\label{pw.weak.conv} Under the conditions of Theorem \ref{strong}, $u_m$ converges to $u$ weakly, pointwise in $t\in \bR$ and uniformly on compact sets. That is, $u_m(t) \rightharpoonup u(t)$ for all $t \in \bR$, and for any $w \in H$, $T_0 \le T \in \bR$,
\[
\lim_{m \rightarrow \infty} |\langle u_m-u,w\rangle|_{L^{\infty}([T_0,T])}=0. 
\]
\end{lemma}
\begin{proof} According to Corollary \ref{equicont}, the conditions described in the statement of Theorem \ref{strong} imply that $\{u_m: m\in {\bf N}\}$ is equicontinuous, hence uniformly equicontinuous on compact sets. Choose $T_0 \le T \in \bR$. Given $\epsilon>0$, choose $\Delta t >0$ so that if $|\delta t| < \Delta t$, $t \in [T_0,T]$, then
\[
\|u_m(t+\delta t)-u_m(t)\| < \epsilon,\,\,m \in {\bf N}; \,\,\|u(t+\delta t)-u(t)\| <\epsilon,
\]
which implies that
\[
\left\|u_m(t) - \frac{1}{2\Delta t}\int_{t-\Delta t}^{t+\Delta t} u_m\right\| < \epsilon, \,\,m\in {\bf N}; \left\|u(t) - \frac{1}{2\Delta t}\int_{t-\Delta t}^{t+\Delta t} u\right\| < \epsilon.
\]
However, according to Lemma \ref{weak.conv}, for any $w \in H$,
\[
\frac{1}{2\Delta t}\int_{t-\Delta t}^{t+\Delta t}\langle u_m-u,w\rangle = \int \left\langle u_m-u,w \frac{1}{2\Delta t}{\bf 1}_{[t-\Delta t,t+\Delta t]}\right\rangle \rightarrow 0,\,\, m \rightarrow \infty.
\]
Therefore, assuming without loss of generality that $\|w\|=1$,
\[
|\langle u_m(t)-u(t),w\rangle | \le 3\epsilon
\]
for $m$ sufficiently large, $t \in [T_0,T]$. Since $T_0, T,$ and $\epsilon>0$ are arbitrary, the proof is complete.
\end{proof}

We require one more fact about the memory operators $R$:

\begin{lemma}\label{rconv} Suppose that $\{Q_m: m\in \bN\} \subset L^1(\bR, {\cal B}_{\rm symm}(H)) \cap C^0(\bR_+, {\cal B}(H))$ defines a sequence of operators $R_m: \lph \rightarrow \lph$ converging strongly to zero. Let ${\cal K} \subset \lph$ be bounded, and assume that all $u\in {\cal K}$ are causal, with common support ${\rm supp}\,u \subset [T_0,\infty)$. Then for any sequence $\{u_m: m \in \bN\} \subset {\cal K}$, and any $v \in \llh $, any $T \in \bR$,
\[
\int_{-\infty}^T\,dt\,\langle v(t), R_m[u_m](t) \rangle \rightarrow 0,\,\, m \rightarrow \infty
\]
\end{lemma}

\begin{proof} 
\begin{eqnarray*}
\int_{-\infty}^T\,dt\,\int_{-\infty}^t \,ds\,\langle v(t), R_m[u_m](t) \rangle & = & \int_{-\infty}^T\int_{T_0}^t\,ds\,\langle v(t),Q_m(t-s) u_m(s) \rangle\\
& = & \int_{-T}^{-T_0} d\sigma \int_{-T}^{\sigma}\,d\tau\, \langle u_m(-\sigma), Q_m(\sigma-\tau) v(-\tau)\rangle\\
& = & \int \,dt\,\langle {\bf 1}_{[-T,-T_0]}(t)\check{u}_m(t),R_m[{\bf 1}_{[-T,-T_0]}\check{v}] (t)\rangle.
\end{eqnarray*}
Note that we have used the symmetry of $Q_m$, and it is for this reason that the assumption was introduced. According to the hypothesis, $\|R_m[w]\|_{L^2((-\infty,t])} \rightarrow 0$ as $m \rightarrow \infty$ for any $w \in \lph, t \in \bR$. Setting $t=-T_0$, $w={\bf 1}_{[-T,-T_0]}\check{v}$. Since ${\cal K} \in \lch$ is bounded, the sequence
\[
\{\|{\bf 1}_{[-T,-T_0]}(t)\check{u}_m\|_{\lh} \} \subset \bR
\]
is bounded, whence the conclusion follows via the Cauchy-Schwarz inequality.
\end{proof}

\begin{proof}  (of Theorem \ref{strong}) We will use repeatedly the algebraic identity: for $K,L \in {\cal B}(H), v,w \in H$,
\[
\langle v,Kv\rangle - \langle w,Lw\rangle = \langle v-w,K(v-w)\rangle + \langle v-w,Lw \rangle
\]
\[
+\langle w,L(v-w) + \langle v,(K-L)w\rangle +\langle w,(K-L)(v-w)\rangle.
\]
If $K$ and $L$ are symmetric, this identity simplifies to
\[
\langle v,Kv\rangle - \langle w,Lw\rangle = \langle v-w,K(v-w)\rangle + 2\langle v-w,Lw \rangle + \langle 2v-w,(K-L)w \rangle
\]
Application of the second version of this identity with $v=u_m,w=u,K=A_m,L=A$ yields
\begin{eqnarray}\label{sc1}
  \langle u_m, A_m u_m \rangle - \langle u,Au \rangle
  & = & \langle u_m - u, A_m (u_m - u) \rangle \nonumber \\ 
  & & + \langle 2 u_m - u, (A_m - A) u \rangle 
        + 2 \langle u_m - u, A u \rangle .
\end{eqnarray}
The right-hand side $f$ in the formal evolution equation \eqref{formeq} for both $u_m$ and $u$ vanishes for sufficiently large negative $t$, else $u$ could not be causal, but then $u_m$ and $u$ must vanish on a common (negative) half-axis, thanks to Corollary \ref{unique1}. The energy identity \eqref{eident} implies that 
\[
\langle u_m, A_m u_m \rangle - \langle u,Au \rangle (t) 
\]
\[
= - \int_{-\infty}^t [\langle B_m u_m,u_m \rangle - \langle B u, u \rangle + \langle R_m [u_m],u_m \rangle - \langle R[u], u \rangle - \langle f,u_m-u \rangle].
\]
Application of the algebraic identities stated above shows that the right-hand side is 
\begin{eqnarray}\label{sc2}
& = & - \int_{-\infty}^t [\langle u_m - u, B_m (u_m - u) \rangle + \langle 2 u_m - u, (B_m - B) u \rangle + 2 \langle u_m - u, B u \rangle \nonumber\\
& & \langle u_m - u, R_m [u_m - u] \rangle + \langle u_m,R_m[u]-R[u] \rangle + \langle u, R_m[u_m-u]-R[u_m-u] \rangle \nonumber \\
& & +  \langle u_m - u, R[u]\rangle + \langle u, R[u_m-u] \rangle 
 + \langle f,u_m-u \rangle].
\end{eqnarray}
Identities \eqref{sc1} and \eqref{sc2} combine to yield
\begin{equation}\label{sc3}
\langle u_m-u,A_m (u_m-u)\rangle(t)
= -\int_{-\infty}^t \langle B_m (u_m-u)+ R_m[u_m-u],u_m-u\rangle + g_m(t),
\end{equation}
in which $g_m \in C^0(\bR)$ is defined by
\begin{eqnarray}\label{sc4}
g_m(t) & = & - \langle 2 u_m(t) - u(t), (A_m - A) u(t) \rangle - 2 \langle (u_m - u)(t), A u(t) \rangle \nonumber\\
& & - \int_{-\infty}^t [\langle 2 u_m - u, (B_m - B) u \rangle + 2 \langle u_m - u, B u \rangle \nonumber\\
& & + \langle u_m,R_m[u]-R[u] \rangle + \langle u, R_m[u_m-u]-R[u_m-u] \rangle \nonumber \\
& & +  \langle u_m - u, R[u]\rangle + \langle u, R[u_m-u] \rangle +  \langle f,u_m-u \rangle].
\end{eqnarray}
Since the $B_m$'s are uniformly bounded operators on $H$ and the $R_m$'s are uniformly bounded operators on $L^2((-\infty,t],H)$ for every $t \in \bR$ (with norm independent of $t$), \eqref{sc3} implies that
\begin{equation}\label{sc5}
\|u_m - u\|^2(t) \le C\langle u_m-u,A_m (u_m-u)\rangle (t) \leq  C \int_{-\infty}^t \|u_m-u\|^2
+ |g_m(t)|.
\end{equation}
in which $C$ means something different each time it occurs, as before, but depends on the quantities indicated the second section.

Select $T_0$ for which $u(t)=u_m(t)=0$ for all $t < T_0, m\in {\bf N}$. Choose $T  \in \bR$. Application of Gronwall's inequality to \eqref{sc5} yields, for $T_0\le t \le T$ and C depending on T along with everything else,
\begin{equation}\label{sc7}
\|u_m-u\|^2(t) \le C_{T-T_0} |g_m|(t).
\end{equation}
in which $C_{T-T_0}$ depends on $T-T_0$ in addition to the standard dependencies.

It remains to see that $g_m(t) \rightarrow 0$ uniformly in $t \in [T_0,T]$, as $m \rightarrow \infty$. The first term in \eqref{sc4} tends to zero pointwise (in $t$) thanks to  the assumption that $A_m \rightarrow A$ strongly, and to the energy estimate (Theorem \ref{existence}) which assures that $\|u_m-u\|$ is bounded uniformly in $m \in \bN$ and $t \in [T_0,T]$. On the other hand, Corollary \ref{equicont} and the uniform bounds on $\{A_m: m \in \bN\}$ implied by the assumption that $\{A_m,B_m,Q_m): m \in \bN\} \subset {\cal P}(C_*,C^*,C_B,C_Q)$ in turn imply that the first term in \eqref{sc4} defines an equicontinuous sequence of continuous functions on $[T_0,T]$. Since any convergent subsequence converges pointwise to the zero function, so does the entire sequence, and uniformly. The second term in \eqref{sc4} tends to zero, uniformly on $[T_0,T]$, thanks to Lemma \ref{pw.weak.conv}, the third because $B_m \rightarrow B$ strongly, the fourth, seventh, eighth and ninth because of Lemma \ref{weak.conv} and Theorem \ref{existence}, the fifth because $R_m \rightarrow R$ strongly in $\lph$, the sixth because of Lemma \ref{rconv}.

\end{proof}

\noindent {\bf Remark.} This result is sharp, in the sense that nothing stronger than continuity can be expected without additional constraints on the various components of the formal evolution problem \eqref{formeq}. In particular, the modulus of continuity cannot be uniform in the right-hand side ($f \in \lh$), even locally.

For example, the 1D linear advection problem
\[
\left(\frac{1}{c}\frac{\partial u}{\partial t}-\frac{\partial u}{\partial x}\right)(t,x)=f(t,x)
\]
conforms to the setting described above, with $H=L^2(\bR)$. The operator coefficients are: $A$ =  multiplication by the positive constant $1/c$, $P = \partial / \partial x$, skew-adjoint with domain $V=H^1(\bR)$, and $B\equiv 0, Q \equiv 0$. For any $f \in L^2(\bR^2)$ ($\equiv \lh$ by Fubini's Theorem), the causal weak solution is
\begin{equation}\label{linadvsol}
u[c,f](t,x)=c\int_{-\infty}^t f(\tau,x+c(t-\tau)) d\tau,
\end{equation}
in which we have explicitly indicated the dependence of the weak solution on the coefficient $1/c$ and the right-hand side $f$.

Suppose $\chi \in C_0^{\infty}(\bR)$, ${\rm supp}\,\chi \subset [-1,1]$, and 
\[
\int \chi  =1.
\]
For $\epsilon>0$, set $f_{\epsilon}(t,x)=\cos((x+t)/\epsilon)\chi(x+t)\chi(x)$. Then $u[1,f_{\epsilon}](t,x)=\cos((x+t)/\epsilon)\chi(x+t)$ for $t>1$, whereas integration by parts shows that
\[
u[c,f_{\epsilon}](t,x)=O\left(\frac{\epsilon}{|c-1|}\right)
\]
for $c \ne 1$. Thus the modulus of continuity of $(c,f_{\epsilon}) \mapsto u[c,f_{\epsilon}](t,\cdot) \in H$ (for $t>1$) cannot be uniform over the bounded set $[a,b] \times \{f_{\epsilon}: \epsilon>0\} \subset \bR \times \lh$, if $1 \in (a,b)$, and in particular this map is not locally uniformly continuous.
The heuristic reason is that $c$ is the wave speed, so changing $c$ changes the position of arriving waves. This position shift has unboundedly large impact if the frequency of oscillation in the solution is allowed to become arbitrarily large. Oscillations in the data in time translate into oscillations in the solution in space, of similar frequency, because the problem is hyperbolic.

On the other hand, additional regularity {\em in time} of the right-hand side entails more regular behaviour of the solution, as one might expect since such regularity damps temporal frequencies.

\begin{corollary}\label{regstrong} In addition to the hypotheses of Theorem \ref{strong}, assume that $f \in H^k_{\rm loc}(\bR,H)$, $k \ge 1$. Then for any choice of $T_0 \le T \in \bR$,
\[
\lim_{m \rightarrow \infty}\|u_m-u\|_{C^{k-1}([T_0,T],H)} = 0.
\]
\end{corollary}
\begin{proof} Follows directly from Theorems \ref{reghyp} and \ref{strong}.
\end{proof}

\begin{proof} of Theorem \ref{smooth}: The meaning of \eqref{gateaux} is that for any $\phi \in C_0^{\infty}(\bR,V)$, 
\begin{eqnarray}\label{sm1}
\int  \langle \delta u, A \phi' + P \phi - B^\ast \phi - R^\ast [\phi] \rangle
& = &\int \langle \delta A u' + \delta B  u + \delta R[u], \phi \rangle \nonumber \\
& = &-\int \langle u, \delta A \phi' - \delta B^* \phi - \delta R^*[\phi] \rangle.
\end{eqnarray}
On the other hand, both $u$ and $u_h$, $h>0$, satisfy \eqref{weakde} with the same right-hand side, so
\begin{eqnarray}\label{sm2}
0 
& = &  \frac{1}{h} \left(\int  \langle u_h, A_h \phi' + P \phi - B_h^\ast \phi - R_h^\ast [\phi] \rangle\right. \nonumber \\
& & \left. - \int  \langle u, A \phi' + P \phi - B^\ast \phi - R^\ast [\phi] \rangle\right) \nonumber \\
& = & \int  \langle u_h, \delta A \phi' - \delta B^\ast \phi - \delta R^\ast [\phi] \rangle \nonumber \\
& & + \int  \left\langle  \frac{u_h-u}{h}, A \phi' + P \phi - B^\ast \phi - R^\ast [\phi] \right\rangle.
\end{eqnarray}
Subtracting \eqref{sm1} from \eqref{sm2} and rearranging, obtain
\[
\int  \left\langle  \left(\frac{u_h-u}{h}-\delta u\right), A \phi' + P \phi - B^\ast \phi - R^\ast [\phi] \right\rangle
=\int  \langle u_h-u, \delta A \phi' - \delta B^\ast \phi - \delta R^\ast [\phi] \rangle
\]
\begin{equation}\label{sm3}
= -\int \langle \delta A (u_h-u)'+\delta B (u_h-u) + \delta R[u_h-u], \phi \rangle.
\end{equation}
In view of equation \eqref{sm3}, the Newton quotient remainder
\[
\frac{u_h-u}{h}-\delta u
\]
is the weak solution of \eqref{formeq} with right-hand side
\[
\delta A (u_h-u)'+\delta B (u_h-u) + \delta R[u_h-u] \in \llh.
\]
In view of Corollary \ref{regstrong} and the energy estimate (Theorem \ref{existence}) imply that
\[
\|\delta A (u_h-u)'+\delta B (u_h-u) + \delta R[u_h-u]\|_{L^2((-\infty,T],H)} \rightarrow 0
\]
as $h\rightarrow 0$ for any $T \in \bR$. The conclusion then follows from anther use of Theorem \ref{existence}. 
\end{proof}

\noindent {\bf Remark.} This result is also sharp, in the sense that the right-hand side must have at least one square-integrable derivative in $t$, if only additional regularity in $t$ is to be imposed. For example, the solution \eqref{linadvsol} of the linear advection equation presented above may be rewritten as
\[
u(t,x) = \int_x^{\infty} f\left(t+\frac{x-y}{c},y\right)\,dy,
\]
from which it is straightforward to see that no less regularity in $t$ will do. On the other hand, the expression \eqref{linadvsol} suggests that additional regularity in $x$ might also support differentiable dependence on $c$. However this conclusion rests on a special feature of the example problem, namely that it admits a propagation of singularity principle (and indeed solution via the method of characteristics, an even more special property). Propagation of singularities along bicharacteristics holds for symmetric or strictly hyperbolic systems with smooth coefficients (see for example  \cite{Tay:81}), and to some limited extent for systems with less regular coefficients \cite[]{BealsReed:82,BealsReed:84,Symes:86b,lew91,BaoSy:95}. Stronger regularity results for dependence on coefficients follow for some of these systems. 

\begin{proof} of Theorem \ref{regfwd}: we present the case $k=2$. The general case follows by an induction argument, which we omit.

As usual, denote by $T_0 \in \bR$ a lower bound for ${\rm supp}\,f$.
In the notation of the proof of Theorem \ref{smooth}, $u \in C^2(\bR,H) \cap C^1(\bR,V)$, thanks to Theorem \ref{reghyp}. Therefore the right-hand side of  of \eqref{gateaux} has a locally square-integrable derivative, whence the causal weak solution $\delta u$ actually belongs to the class $C^1(\bR,H) \cap C^0(\bR,V)$. Applying Theorems \ref{existence} and \ref{reghyp} repeatedly, one sees that $\delta u$ satisfies for any $T_1 \in \bR$
\[
\|\delta u\|_{L^{\infty}([T_0,T_1],V)} \le C_{T_1-T_0} \|(\delta A, \delta B, \delta Q)\|_{\cal P} \|f\|_{H^2([T_0,T_1])}.
\]
It follows that the linear map $D{\cal F}[(A,B,Q)]: {\cal M} \rightarrow C^0([T_0,T_1],W)$ defined by 
\[
D{\cal F}[(A,B,Q)](\delta A, \delta B, \delta Q) = S \delta u
\]
is continuous.

Suppose that $(A_m,B_m,Q_m) \rightarrow (A,B,Q)$ in norm (of ${\cal M}$). Denote by $\delta u_m$ the solution of \eqref{gateaux} with $(A,B,Q)$ replaced by $(A_m,B_m,Q_m)$. It follows from the definition (equation \eqref{gateaux}) that $\delta u_m - \delta u$ is the (strong) solution of
\[
A(\delta u_m -\delta u)' + P(\delta u_m-\delta u) + B (\delta u_m-\delta u) + R[\delta u_m-\delta u] 
\]
\begin{equation}\label{gatgat}
= -[(A_m-A) \delta u_m + (B_m-B) \delta u_m + R_m[\delta u_m]-R_m[\delta u]  + \delta A (u_m-u)' + \delta B  (u_m-u) + \delta R[u_m-u]].
\end{equation}
Theorem \ref{reghyp} implies that $\{\delta u_m:m\in \bN\}$ is a bounded set in $C^1((-\infty,T_1],H)\cap C^0((-\infty,T_1],V)$, whence in $H^1_{\rm loc}((-\infty,T_1],H)$ also; moreover ${\rm supp}\,\delta u_m \subset [T_0,\infty)$ for all $m\in \bN$. It follows that the first three terms on the right-hand side of \eqref{gatgat} tend to zero in $H^1([T_0,T_1],H)$.  Applying Theorem \ref{strong} to $u_m'-u'$, the difference of solutions of equations of the form (\ref{formeq}) with the same causal right-hand side $f' \in H^1_{\rm loc}(\bR,H)$, one sees that $u_m-u \rightarrow 0$ in $H^1([T_0,T_1],H)$ also, whence the second three terms also tend to zero in this sense. Thus the entire right-hand side of \ref{gatgat} tends to zero as $m \rightarrow \infty$ in the sense of $H^1([T_0,T_1],H)$. Now it follows from Theorem \ref{reghyp} that
\[
(D{\cal F}[(A_m,B_m,Q_m)]-D{\cal F}[(A,B,Q)])(\delta A, \delta B,\delta Q) = S(\delta u_m-\delta u) \rightarrow 0
\]
(in the sense of $C^0([T_0,T_1],W)$. That is, $D{\cal F}: {\cal P} \times {\cal M} \rightarrow C^0[T_0,T_1],W)$ is continuous, and the theorem is proved in the case $k=2$.

\end{proof}

\section{Symmetric Hyperbolic Systems: continuous dependence,  proofs of Theorems \ref{symmhyp:wellposed} - \ref{symmhyp:speed}}

For convenience, we repeat the key definitions from the second section.

Recall that $k\times k$ symmetric hyperbolic systems take the form
\begin{equation}
\label{symmhyp:repeat}
a \frac{\partial u}{\partial t} + p(\nabla) u + bu + q \ast u = f\mbox{ in }\Omega \times \bR; \, u = 0 \mbox{ for } t < 0,
\end{equation}
in which the coefficient matrices $a, b, $ and $q$ are $k\times k$, and the $k \times k$ matrix differential operator in the ``space'' variables $\bx \in \Omega$ has symmetric and constant coefficient matrices.
 
The set of admissible coefficients $M \subset L^{\infty}(\Omega,\bR^{k \times k}_{\rm symm}) \times L^{\infty}(\Omega, \bR^{k\times k} ) \times L^1(\bR_+,L^{\infty}(\Omega,\bR^{k \times k}_{\rm symm} ))$ is defined by 
\[
M(C_*,C^*,C_B,C_Q) = \{(a,b,q): C_*I < a(\bx) < C^*I \mbox{ for all }\bx \in \Omega, \|b\|_{L^{\infty}(\Omega, \bR^{k\times k} )} < C_B, 
\]
\[
\|q\|_{L^1(\bR_+,L^{\infty}(\Omega,\bR^{k \times k}_{\rm symm} ))} < C_Q, q(t) = 0 \mbox{ for }t<0\},
\]
\begin{equation}
\label{symmhyp:admissible}
M = \bigcup \{ M(C_*,C^*,C_B,C_Q): 0 < C_* \le C^*, C_B, C_Q \in \bR_+\}.
\end{equation}
$M$ corresponds to a collection of problems of the form \eqref{symmhyp:repeat}, with common $p(\nabla)$.

The Hibert space of states is $H=L^2(\Omega)^k$. Provide $p(\nabla)$ with a dense domain $V \subset H$, and assume the $p(\nabla):V \rightarrow H$ is skew-adjoint, and that the norm in $V$ is equivalent to the graph norm of $p(\nabla)$.

The hypotheses of Theorem \ref{existence} are all satisfied, so

\begin{corollary}\label{symmhyp:weak} Suppose that $(a,b,q) \in M$, and $f \in \llh$ is causal. Then there exists a unique weak solution of \eqref{symmhyp:repeat}, satisfying for every $\phi \in H^1(\bR,L^2(\Omega)^k)$ 
\[
\int_{-\infty}^{\infty}\,dt\,\int_{\Omega}\,d\bx\,[u(-A \phi'-p(\nabla)\phi + b\phi+q\ast \phi) - f \phi] = 0.
\]
There exists an increasing function $C: \bR_+ \rightarrow \bR_+$, depending on $C_*, C^*$, and bounds for $b,q$, so that for every $t \in \bR$,
\begin{equation}
\label{symmhyp:weakenergy}
\|u(\cdot,t)\|^2_{L^2(\Omega)^k} \le C(t) \int_0^t d\tau 
\|f(\cdot,\tau)\|^2_{L^2(\Omega)^k}.
\end{equation}
\end{corollary}

Theorems \ref{symmhyp:wellposed} and \ref{symmhyp:tracereg} now follow directly from Theorems \ref{smooth} and \ref{regfwd} respectively.

Multiplication by a member of $L^{\infty}(\Omega)$ defines a continuous map from $L^{\infty}(\Omega)$ to ${\cal B}(L^2(\Omega))$, so convergence of the coefficients $a,b,q$ in  $L^{\infty}(\Omega) $ is sufficient to induce uniform, hence strong, convergence of the corresponding operators, hence convergence of the solutions per Theorem \ref{strong}. However, convergence in a weaker sense is also sufficient to induce strong operator convergence. The key observation is the following result, identical to Lemma 2.8.5 in \cite[]{stolk}:

\begin{lemma}\label{mult} Let $(\Omega,{\cal A},\mu)$ be a measure space, $\{r_m\} \subset L^{\infty}(\Omega,\mu)$ with $\|r_m\|_{L^{\infty}(\Omega,\mu)} \le R \in \bR_+$, $\{f_m\} \subset L^2(\Omega,\mu)$ with $\|f_m\|_{L^2(\Omega,\mu)} \le F \in \bR_+$ for all $m \in {\bf N}$. Suppose that $r_m \rightarrow 0$ in $\mu$-measure. Then for any $g \in L^2(\Omega,\mu)$, 
\begin{equation}\label{cim-lim}
{\rm lim}_{m\rightarrow \infty} \int_E r_m f_m g d\mu = 0.
\end{equation}
\end{lemma}

\begin{proof}
Suppose on the contrary that such sequences $\{r_m\}, \{f_m\}$ and square-integrable $g$ exist, also an $\eta >0$, for which the left-hand side of \eqref{cim-lim} remains $\ge \eta$ along a common subsequence. Without loss of generality, renumber the subsequence so that 
\begin{equation}\label{absurd}
\left|\int_E r_m f_mg d\mu \right| \ge \eta, \,\,m\in {\bf N}.
\end{equation}
Convergence in measure of $\{r_m\}$ means that for any $\epsilon > 0$,
\[
\mu[E_{\epsilon}(r_m)] \rightarrow 0 \,\,{\rm as}\,\,m \rightarrow \infty, \,\,{\rm where} \,\,E_{\epsilon}(r_m)=\{\bx \in E: |r_m(\bx)| \ge \epsilon\}.
\]
Choose $\epsilon$ so that $\epsilon F \|g\|_{L^2(\Omega,\mu)} < \eta/2$. 

From this definition and the Cauchy-Schwarz inequality, one sees that
\[
\left|\int_E r_m f_mg d\mu \right| \le \epsilon\int_{E\setminus E_{\epsilon}(r_m)} |f_m g| d\mu+R\int_{E_{\epsilon}(r_m)} |f_m g| d\mu
\]
\[
\le \epsilon F \|g\|_{L^2(E,B,\mu)} + RF\left(\int_{E_{\epsilon}(r_m)} g^2 d\mu\right)^{\frac{1}{2}}
<\frac{\eta}{2} + RF\left(\int_{E_{\epsilon}(r_m)} g^2 d\mu\right)^{\frac{1}{2}}
\]
By passing if necessary to a further subsequence, we may assume that
\[
\mu[E_{\epsilon}(r_m)] \le 2^{-m}\,\, \Rightarrow \,\, \sum_m \mu[E_{\epsilon}(r_m)]  < \infty.
\]
Thus the characteristic functions of the sets $E_{\epsilon}(r_m)$ are almost everywhere convergent to zero as $m \rightarrow \infty$. Since $|g|^2 \in L^1(\Omega,\mu)$, it follows from the Lebesgue Dominated Convergence Theorem that for large enough $m$, 
\[
\left( \int_{E_{\epsilon}(r_m)} g^2 d\mu \right)^{\frac{1}{2}} < \frac{\eta}{2 R F}
\]
Thus the left-hand side of \eqref{absurd} can be made smaller than $\eta$, a contradiction.
\end{proof}

\begin{lemma}\label{mult-cim} Suppose that $\{a_m\}_{m=1}^{\infty} \subset L^{\infty}(\bR^n)$ converges in measure to $a \in L^{\infty}(\bR^n)$, and that $A_m, m \in {\bf N}$ and $A\in {\cal B}(L^2(\bR^n))$ are defined by  
\[
(A_m u) (\bx) = a_m(\bx) u(\bx),\,\,(A u) (\bx) = a(\bx) u(\bx),\,\,u \in L^2(\bR^n).
\]
Then $A_m \rightarrow A$ strongly. The same is true for similar sequences of operators on $L^2(\bR^n)^p$ defined by sequences of $p \times p$ matrix-valued functions whose components converge in measure.
\end{lemma}

\begin{proof} In fact, the operators so defined are self-adjoint, and  
\[
\|A_m u - A u\|^2 = \int (a_m-a)[(a_m-a)u] u \rightarrow 0,
\]
as follows from Lemma \ref{mult}, taking $a_m-a$ for $r_m$, $(a_m-a)u$ for $f_m$, and $u$ for $g$ in the notation of that lemma.
\end{proof}
 
\begin{proof} of Theorem \ref{symmhyp:strong}: follows immediately from Lemma \ref{mult-cim} and Theorem \ref{strong}.
\end{proof}

One way to obtain convergence in measure is via mollification. Let $\{\eta_m: m \in \bN\}$ be a Dirac sequence of mollifiers, as in the proof of Proposition \ref{energy}. Since $a \in L^p_{\rm loc}(\bR^d)$ for any $p$, $\eta_m \ast a \rightarrow a$ pointwise almost everywhere, hence in measure. An application of this observation is the {\em finite speed of propagation} property: for solutions of \eqref{symmhyp:repeat}, support expands at finite speed. This property is well-understood for hyperbolic systems with smooth coefficients: for example, \cite{Lax:PDENotes}, Ch. 4, presents a proof of:

\begin{proposition}\label{std-sm} In addition to the hypotheses of  Theorem \ref{symmhyp:wellposed}, suppose that the coefficients in \eqref{symmhyp:repeat} are smooth (of class $C^{\infty}(\bar{\Omega})$). Then
\begin{itemize}
\item[1. ] the causal weak solution $u \in \llh$ is smooth: $u \in C^{\infty}(\bR^{n+1})^p$, and
\item[2. ] if $\phi \in C^{\infty}(\bR^n)$ satisfies
\[
a + \sum p_i \frac{\partial \phi}{\partial x_i} > 0
\]
and ${\rm supp}(f) \bigcap \{(\bx,t): \phi(\bx)  > t\}=\emptyset$, then $u(\bx,t)=0$ if $\phi(\bx) \le t$.
\end{itemize}
\end{proposition}

\begin{corollary}\label{std-dd} In the setting of Proposition \ref{std-sm}, suppose that $\tau \in \bR$ satisfies 
\begin{equation}\label{dd-bd}
\tau a(\bx)+\sum_{i=1}^n p_i \xi_i  \ge 0,\,\,\bx \in \bR^n,\, |{\bf \xi}|=1.
\end{equation}
If $\bx_0 \in \bR^n, t_0 \in \bR$ satisfy
\[
f(\bx,t) = 0 \,\,{\rm if}\,\, \tau|\bx-\bx_0|+t_0-t \ge 0,
\]
then $u(\bx_0,t_0)=0$.
\end{corollary}

Note that the argument given in \cite[]{Lax:PDENotes} does not quite encompass Proposition \ref{std-sm} and Corollary \ref{std-dd}: it does not apply to systems like \eqref{symmhyp:repeat} with memory terms. However the extension is straightforward.

\begin{proof} of Theorem \ref{symmhyp:speed}: Using a sequence of mollifiers as before, obtain a sequence $\{(a_m,b_m,q_m): m \in \bN\}$ of smooth coefficients converging in measure to $(a,b,q)$. Then the corresponding operators converge strongly to those induced by $(a,b,q)$.  Since \eqref{tau-bd} holds almost everywhere, it follows that $a_m$ satisfies the spectral inequality \eqref{dd-bd} in $\bR^n$ for sufficiently large $m$, whence Corollary \ref{std-dd} implies that $u_m$ vanishes in $\omega$. Since, $u_m \rightarrow u$ in $L^2(\Omega)^k$ the conclusion follows.
 \end{proof}  

\section{An example: viscoelasticicty}
The dynamic equations of linear viscoelasticity may be written as
\begin{eqnarray}\label{ve}
\rho \frac{\partial \bv}{\partial t} & = & \nabla \cdot \sigma + \bff, \nonumber \\
\Gamma \ast \frac{\partial \sigma}{\partial t} & = & \frac{1}{2}(\nabla \bv + \nabla \bv^T). 
\end{eqnarray}
in which $\bv$ is the particle velocity field, $\sigma$ the stress tensor, $\bff$ a body force density, $\rho$ the mass density, and $\Gamma$ the inverse Hooke operator \cite[]{christensen,pipkin}. Viscoelasticity differs from elasticity in that the inverse Hooke operator is a causal convolution operator (in time), rather than a temporally local multiplication operator. It follows from \eqref{ve} that the strain rate (right-hand side of the second equation) is the convolution of the stress with the indefinite time integral of $\Gamma$. 

We will assume instantaneous elastic response: a nonzero strain rate arises immediately from a stress impulse. Under this assumption, the kernel $\Gamma$ can be decomposed as
\[
\Gamma(t)  = \Gamma^e \delta(t)+ \gamma(t),
\]
in which $\Gamma^e$ is the elastic inverse Hooke tensor (inverse of the unrelaxed modulus), and $\gamma$ is a causal kernel. Both the elastic kernel $\Gamma^e$ and the memory kernel $\gamma(t)$ act as spatially-variable, symmetry-preserving linear operators on symmetric tensor fields. The conventional representation of such things by 4-index tensors,
\[
\Gamma^e = \left.\left(\Gamma^e_{ijkl}\right)\right|_{i,j,k,l=1}^3,\,\, \gamma = \left.\left(\gamma_{ijkl}\right)\right|_{i,j,k,l=1}^3, 
\]
thus entail the symmetries
\begin{equation}\label{vesymm}
\Gamma^e_{ijkl}=\Gamma^e_{jikl}=\Gamma^e_{ijlk}=\Gamma^e_{klij},\,\,i,j,k,l=1,2,3,
\end{equation}
and similarly for $\gamma$.

To avoid technical complications, assume that the viscoelastic material occupies all of $\bR^3$. We require that, for some $0 < g_* \le g^*$,
\begin{itemize}
\item[1. ] $\Gamma^e$ is elliptic: for any symmetric $\sigma \in \bR^{3\times 3}$,
\begin{equation}\label{vell}
g_*\|\sigma\| \le \|\Gamma^e(\bx)\sigma\| \le g^*\|\sigma\|, \,\, \bx \in \bR^3;
\end{equation}
\item[2. ] $\Gamma^e \in L^{\infty}(\bR^3,{\cal B}(\bR^{3 \times 3}_{\rm symm}))$;
\item[3. ]$\gamma \in W^{1,1}(\bR, L^{\infty}(\bR^3,{\cal B}(\bR^{3 \times 3}_{\rm symm})))$.
\end{itemize}

For the ``state space'' $H$ of the viscoelastic system we choose $H = L^2(\bR^3,\bR^9) \equiv L^2(\bR^3, \bR^{3 \times 3}_{\rm symm}) \times L^2(\bR^3,\bR^3)$. The inner product in $H$ is defined by
\[
\langle u_1,u_2 \rangle = \int_{\bR^3}\, {\rm tr}\, \sigma_1^T \sigma_2 + \bv_1^T\bv_2,\,\,
u=\left(
\begin{array}{c}
\sigma \\
\bv
\end{array}
\right).
\]
The assumptions 1-3 above and the symmetries \eqref{vesymm} imply that 
\[
Au=\left(
\begin{array}{c}
\Gamma^e \sigma \\
\rho \bv
\end{array}
\right), 
\,\, u=\left(
\begin{array}{c}
\sigma \\
\bv
\end{array}
\right) 
\in H
\]
defines a bounded, self-adjoint positive-definite operator $A \in {\cal B}(H)$. 

Define the differential operator $p(\nabla): C_0^{\infty}(\bR^3,\bR^9) \rightarrow C_0^{\infty}(\bR^3,\bR^9)$ by
\[
p(\nabla)u=-\left(
\begin{array}{c}
\frac{1}{2}(\nabla \bv + \nabla \bv^T)\\
\nabla \cdot \sigma 
\end{array}
\right), 
\,\, u=\left(
\begin{array}{c}
\sigma \\
\bv
\end{array}
\right) 
\in C_0^{\infty}(\bR^3,\bR^9) \equiv C_0^{\infty}(\bR^3, \bR^{3 \times 3}_{\rm symm}) \times C_0^{\infty}(\bR^3,\bR^3).
\]
$p(\nabla)$ is antisymmetric and densely defined in $H$. Define
\begin{itemize}
\item $H^1_{\rm div}(\Omega,\bR^{3 \times 3}_{\rm symm})$ to be the subspace of $L^2(\Omega,\bR_{\rm symm}^{3 \times 3})$ consisting of square-integrable symmetric matrix valued functions, each column of has a square-integrable divergence;
\item $H^1_{\rm free}(\Omega)$ to be the subspace of vectors in  $\sigma \in H^1_{\rm div}(\Omega,\bR^{3 \times 3}_{\rm symm})$ satisfying 
\[
\sigma(\bx)\cdot {\bf n} = 0,\,\,\mbox{a. e. }\bx \in \partial \Omega;
\]
\item $V = H^1_{\rm free}(\Omega)^3 \times H^1(\Omega)^3$.
\end{itemize}
An argument similar to that explained in Appendix A shows that $p(\nabla)$ has a self-adjoint extension $P: V \rightarrow H$, and that the natural norm on $V$ is equivalent to the graph norm of $P$.

Let 
\[
b = \lim_{t\rightarrow 0^+} \gamma(t,\cdot) \in L^{\infty}(\bR^3,{\cal B}(\bR^{3 \times 3}_{\rm symm})),
\]
and 
\[
q= \lim_{t \rightarrow 0^+} {\bf 1}_{[t,\infty)}\frac{\partial \gamma}{\partial t} \in 
L^1(\bR, L^{\infty}(\bR^3,{\cal B}(\bR^{3 \times 3}_{\rm symm}))).
\]
Then
\[
\gamma \ast \frac{\partial \sigma}{\partial t} = b \sigma + q \ast \sigma.
\]
Define $B\in {\cal B}(H)$ and $Q \in L^1(\bR,{\cal B}(H))$ by
\[
Bu= \left(
\begin{array}{c}
b \sigma \\
0
\end{array}
\right),
\,\,
Q(t)u=\left(
\begin{array}{c}
q(t)\sigma \\
0
\end{array}
\right),
\,\, u=\left(
\begin{array}{c}
\sigma \\
\bv
\end{array}
\right) 
\]

Finally, define $f \in \llh$ by $f=(0,\bff)^T$.

With these definitions, the system \eqref{ve} is formally equivalent to the evolution problem \eqref{formeq}. The theory developed here thus assures the existence of weak solutions of \eqref{ve}, in material models including discontinuities of densities and/or elastic moduli and/or relaxation moduli.

An immediate consequence of Theorem \ref{symmhyp:speed} is

\begin{corollary}\label{fps} Denote by $c_p$ the maximum quasi-p-wave velocity of the viscoelastic system \eqref{ve}, defined as 
\[
c_p = \mbox{\rm ess sup} \{\lambda_{\rm max} (\Gamma^e(\bx)[{\bf \xi}{\bf  \xi}^T]/\rho(\bx): \bx, {\bf \xi} \in \bR^3, {\bf \xi}^T{\bf \xi}=1 \}.
\]
Suppose that $(\bx_0, t_0)$ satisfies
\[
|\bx-\bx_0| > c_p(t_0-t)
\]
for every $(\bx,t) \in {\rm supp}\,{\bf f}$. Then the causal weak solution $(\sigma,\bv)$ of \eqref{ve} vanishes in a neighborhood of $(\bx_0,t_0)$.
\end{corollary}

\section{Discussion}
The theory developed in the preceding pages provides a basic framework for inverse problems defined by symmetric hyperbolic systems. However a number of important issues remain to be addressed. We shall describe some of these, and some implications of our theory in each case.

\subsection{Descent directions}
We have shown how to formulate the ``forward map'' (${\cal F}_{f,m}$ in the notation of Theorem \ref{symmhyp:tracereg}) as a continuously differentiable map on an open subset of a suitable Banach space, with values in a (subset of a) suitable Hilbert space. Thus least squares objective functions for such problems are continuously differentiable in a well-defined sense. However, the domain metric specified by the theory is some version of $L^{\infty}$, in the case of concrete problems of form \eqref{symmhyp}, or the operator norm in the case of abstract problems (Theorem \ref{regfwd}). Thus the ambient Banach space is in all cases nonreflexive, and does not have the properties required for a sensible definition of gradient, as explained for instance by \cite{KaltenbacherNeubauerScherzer}. Thus differentiability does not necessarily bring with it natural access to descent directions.

Additional information about the derivative is available, however, via the well-known adjoint state method. For brevity, we explain this construction in the case of the abstract evolution problem \eqref{formeq} in the differential case and without lower-order term, that is, $B, R = 0$, and we proceed formally. With notation as in the statement of Theorem \ref{regfwd}, we define 
\[
{\cal F}_{f,m}:  {\cal B}_{\rm symm}^+(H) \rightarrow L^2([0,T], W)
\]
by
\[
{\cal F}_{f,m}[A] = S[m]u,
\]
where $u \in C^1(\bR,H) \cap C^0(\bR,V) $ is the solution of
\[
A u' + P u = f \in H^k_{\rm loc}(\bR,H), \,\, f(t), u(t)=0, \, t<0
\]
guaranteed by the theory. If $k \ge 2$, then ${|cal F}_{f,m}$ is of class $C^1$, and its derivative is given by 
\[
D{\cal F}_{f,m}[A]\delta A = S[m]\delta u,
\]
where
\[
A \delta u' + P \delta u = -\delta A u' , \,\, \delta u(t)=0, \, t<0
\]
Suppose that we are provided data $d \in L^2([0,T],W)$. The adjoint state method \cite[]{Plessix:06} represents  the derivative of the least-squares function
\[
J_{f,m}[A;d] = \frac{1}{2}\|{\cal F}_{f,m}[A] - d\|^2
\]
 as follows. Suppose that $w$ is a strong solution of
\begin{equation}
\label{adjeq}
A'w + Pw = S[m]^*(d-{\cal F}_{f,m}[A]);\,\, w(t)=0, t > T.
\end{equation}

Then for any $\delta A \in {\cal B}_{\rm symm}(H)$,
\begin{equation}
\label{asm}
DJ_{f,m}[A;d]\delta A = \int_0^T\,dt\,\langle \delta A \,u'(t), w(t) \rangle.
\end{equation}
It is possible to extract a {\em gradient}, that is, a direction of fastest ascent for $J_{f,m}$, from this representation of the derivative. Write (for $u,v \in H$) $u \otimes v$ for the rank 1 member of ${\cal B}(H)$ defined by $u \otimes v (w) = \langle v,w \rangle u$. Such rank-1 operators are very special instances of operators of {\em trace class}. We refer to \cite{Conway:90}, pp. 267-268 and 275 for the properties of this subspace of ${\cal B}(H)$ cited here, in particular these: the product of a bounded operator and a trace-class operator is of trace class; a real-valued linear {\em trace function} ${\rm tr}$ is defined on the trace class, generaizing the trace of matrices. The adjoint-state expression \eqref{asm} is equivalent to 
\begin{equation}
\label{asm-trace}
DJ_{f,m}[A;d]\delta A = \mbox{tr }\left(\delta A \,\int_0^T\,dt\, u'(t) \otimes w(t) \right),
\end{equation}
well-defined since the integral inside the parenthesis defines a trace-class operator. Ths trace of the product defines a duality pairing, expressing the trace class as the ``predual'' of ${\cal B}(H)$. 

Assuming that $f, d$ ars sufficiently smooth in $t$, the integrand in \eqref{asm-trace} is well-approximated by quadrature:
\begin{equation}
\label{asm-approx}
DJ_{f,m}[A;d]\delta A = \mbox{tr }\left(\delta A \,\Delta t\, \sum_{n=0}^{[T/\Delta t]}\, u'(n\Delta t) \otimes w(n\Delta t) \right)
\end{equation}
Linear functionals on ${\cal B}(H)$ of the form given by the right-hand side of \eqref{asm-approx}
form precisely the subset of the dual ${\cal B}(H)^*$ consisting of functionals continuous in the weak operator topology. Since the unit ball of ${\cal B}(H)$ is compact with respect to weak convergence, the functional defined by the right-hand side of \eqref{asm-approx} has a maximizer over the ball, which is a candidate for an approximate gradient of $J_{f,m}$ at $A$.

This observation remains correct if $A$ and $\delta A$ are restricted to subspaces of ${\cal B}(H)$, for example $L^{\infty}$ matrix-valued functions as in the definition of symmetric hyperbolic systems \eqref{symmhyp}. 

The program outlined in the preceding paragraphs leaves a number of details to be filled in. For example, actually {\em finding} the maximizer of the linear functional in \eqref{asm-approx} is in some sense a finite dimensional problem, but not a particularly simple one. Also, the precise relation between right-hand sides of \eqref{asm-trace} and \eqref{asm-approx} remains to be established. Finally, the adjoint state evolution problem \eqref{adjeq} involves a right-hand side ($S[m]^*(d-{\cal F}_{m,f}[A]) \in L^2_{\rm loc}(\bR, V^*)$ outside of the class for which solutions, strong or weak, have been shown to exist. We will address this latter problem in the next subsection.

We also note that the gradient, defined as a direction of fastest ascent, is not smooth, or even necessarily a function, even under hypotheses which guarantee that $J_{f,m}$ has several derivatives. This observation leads to various developments in convex optimization, beyond the scope of this paper. We note merely that \cite{BamChavLai:79}, a landmark early paper on inverse problems in wave propagation, addressed this point explicitly.

\subsection{Sharpness of trace regularity and source definition}

Theorem \ref{symmhyp:tracereg} gives sufficient conditions for the trace of a solution of the symmetric hyperbolic system to be well-defined and depend smoothly on the coefficients and right-hand side in \eqref{symmhyp}.
However these conditions are only sufficient, not necessary. Under some conditions, the trace of even a weak solution may be well-defined and depend continuously on the coefficients and right-hand side. 

For example, suppose that the coefficients are smooth in a neighborhood $\omega$ of the boundary $\partial \Omega$ and containing the measurement surface $\Gamma$, that the right-hand side is supported outside of this neighborhood, and that no rays of geometric optics originating outside this neighborhood are tangent to the boundary.  Then microlocal propagation of regularity shows that the the map $f \in H^1_0([0,T],L^2(\Omega \setminus \omega)^k) \rightarrow L^2([0,T], L^2(\Gamma)^l)$ extends continuously to a map $L^2([0,T],L^2(\Omega \setminus \omega)^k) \rightarrow L^2([0,T], L^2(\Gamma)^l)$. See for example \cite[]{Symes:83,Lasi:86}.

By duality, sources supported on $\Gamma$ also give rise to well-defined solutions with estimates similar to those proven in the body of this paper. Combining these two observations leads to definition of variants of the {\em Dirichlet-to-Neumann map} \cite[]{Uhl18} which forms the natural abstraction of ``data'' for many inverse problems. These mappings, from suitable (source) data supported on $\Gamma$ to other data supported on $\Gamma$ ( or on another similar surface), thus inherit a definition and regularity properties for minimally regular coefficient classes as described above. 

\subsection{Nonphysical extension via operator coefficients} 
Perhaps surprisingly, the results on systems with operator coefficients (Theorems 5 - 9) seem likely to be useful in themselves. This utility originates in the resistance of inverse problems for hyperbolic systems to the least-squares (or least-error) approaches that have successfully treated many other types of inverse problems in science and engineering. Natural objective functions change rapidliy in some directions, not in others, and appear to possess many stationary points far from any useful model estimate, a nearly fatal feature for variants of Newton's method, often the only feasible approach to model estimation \cite[]{GauTarVir:86,santsym:book}. 

The last-named author has suggested that use of operator coefficients as an {\em extension} of normal continuum physics could conceivably convexify the normal least-squares objective \cite[]{geoprosp:2008}. Far from an arbitrary introduction of additional degrees of freedom, this reformulation of the inverse problem is closely related to well-tested ideas in seismic data processing, and has already a partial theoretical justification \cite[]{StolkDeHoopSymes:09,ShenSymes:08}. Very recently, implementations of extended least squares inversion using operator coefficients have produced very promising early results \cite[]{BiondiAlmomin:SEG12} which tend to support the convexification hypothesis.

\subsection{And so on...}
The fundamental model problem \eqref{symmhyp} does not include static differential constraints, so the results established above do not apply directly to Maxwell's equations, for example. We suspect that the theory can be extended to accommodate such problems. 

Imposition of additional regularity requirements on coefficient matrices, beyond boundedness and measurability, leads to existence of well-behaved solutions with less regularity. \cite{stolk} offers some results concerning second-order problems with additional regularity of both coefficients and data. 

Both physical heuristics and numerical evidence for inverse problems in seismology \cite[]{VirieuxOperto:09} suggest strongly that the mapping ${\cal F}$ is more linear when the right-hand side is smoother in time. Estimates similar to those developed in the proof of Theorem 9 show that for right-hand sides of the form $f(\epsilon t)$, the Newton quotient remainder is $O(h \epsilon)$ relative to the directional derivative, for any $(\delta A, \delta B, \delta Q)$ not in the null space of $D{\cal F}$. This observation can be developed into a theoretical justification of the frequency continuation strategy employed in every successful contemporary algorithm for least squares inversion of seismic data.

\section{Conclusion}

This paper has demonstrated that the relation between coefficients and solutions of symmetric hyperbolic integro-differential systems has many of the regularity properties assumed (usually without comment) in the applied literature on inverse problems in wave propagation, while allowing for a degree of irregularity in the spatial dependence of coefficients which seems adequate to model the heterogeneity observed in real materials. The methods of proof are classical, dating to the middle of the last century, yet so far as we are aware these results have not previously been explicitly formulated or collected in one place. Several by-products of our development are interesting in themselves: finite speed of propagation for waves in highly heterogeneous materials, a theoretical framework for the systems with operator coefficients which have recently arisen as a key element of convexifying relaxations for waveform inversion, and the clear necessity of developing descent algorithms for non-reflexive Banach metrics, for example. We have also pointed out a number of directions in which the theory suggests extensions, or poses open questions. Indeed, we expect that our results will serve as an initial step in the development of a solid mathematical foundation for inverse problems in heterogeneous media.

\section{Acknowledgement}
KDB acknowledges the support of his work by the Rice University VIGRE program, funded by the National Science Foundation.
The work of KDB and WWS was supported in part by the sponsors of The Rice Inversion Project.
The research of CCS. is supported by the Netherlands Organisation  
for Scientific Research, through VIDI grant 639.032.509. The authors are grateful to Yin Huang for a very careful reading of a preliminary version and for several helpful suggestions. WWS acknowledges the Isaac Newton Insitute for Mathematical Sciences at Cambridge University (programme on Inverse Problems, fall 2011) for the opportunity to take advantage of its superbly productive environment during the final stages of our work on this project.

\bibliographystyle{seg}
\bibliography{master,masterref}

\append{the skew-adjoint property of the acoustic grad-div operator}\label{append:skew}
We establish skew-adjointness of the operator D, defined formally in \eqref{grad-div}, by appealing to an auxiliary result, similar to necessary and sufficient conditions for self-adjointness found in many texts on functional analysis (for example, \cite{Conway:90}).

\begin{lemma}\label{abstr-skew}
Suppose that $H$ is a Hilbert space, $V \subset H$ a dense subspace, and $L:V \rightarrow H$ a skew-symmetric linear operator. Then $L$ is skew-adjoint if and only if
for any $c \in \bR \setminus \{0\}$, $L+cI$ is surjective.
\end{lemma}
\begin{proof}
Suppose first that $L$ is skew-symmetric and $L+cI$ is surjective for every $c \in \bR \setminus \{0\}$. Since the domain of $L$ is dense, for any such $c$,
\[
{\rm ker}(L^*+cI) = {\rm rng}(L+cI)^{\perp} = \{0\}.
\]
Choose $x \in {\cal D}(L^*)$, and $c \in \bR, c \ne 0$. Then there must be $y \in V$ so that $(L-cI)y =  (L^*+cI)x$. However, since $L$ is skew-symmetric, $V \subset {\cal D}(L^*)$ so $(L-cI)y = -(L^*+cI)y=(L^*+cI)x$. However, we just saw that $L^*+cI$ is injective, so $x = -y \in V$. Thus $V = {\cal D}(L^*)$ whence $L$ is skew-adjoint.

Conversely, if $L$ is skew-adjoint, there can be no nonzero  solutions to $(L+cI)x=0$ with nonzero $c$, as follows from skew-symmetry. Hence there are no nonzero solutions to $(L^*+cI)x=0$ with nonzero $c$, since $L=L^*$. That is, the kernel of $L^*+cI$ is trivial for nonzero $c$, whence the range of $L+cI$ is dense for any nonzero $c$. However $L$ is closed (see for instance \cite{Conway:90}, Proposition X.1.6, p. 305), so its range is closed (this is proved just as is a similar fact for closed symmetric operators, see \cite{Conway:90}, Proposition X.2.5, p. 310) thus ${\rm rng}(L+cI)=H$.
\end{proof}

\begin{proposition}\label{acoustics-skew}
The differential operator $p(\nabla)$ with domain $C^{\infty}_0(\Omega)^4$, defined formally by \eqref{grad-div}, extends to a skew-adjoint operator $P$ on $H$ with dense domain 
\[
V = H^1_0(\Omega) \times H^1_{\rm div}(\Omega).
\]
\end{proposition}

\noindent {\bf Remark:} The Hilbert space $H^1_{\rm div}(\Omega)$ is the dense subspace of $L^2(\Omega)^3$ obtained by completing $C^1(\Omega)^3$ in the graph norm of the divergence operator. Equivalently, $\bv \in L^2(\Omega)^3$ belongs to $H^1_{\rm div}(\Omega)$ if and only if there exists $C \ge 0$ so that for every $\phi \in H^1_0(\Omega)$, 
\[
|\langle \nabla \phi, \bv \rangle_{L^2(\Omega)^3}| \le C\|\phi\|_{L^2(\Omega)}.
\]

\begin{proof}
From the definitions, $p(\nabla)$ extends to $P:V \rightarrow L^2(\Omega)^4$. According to Lemma \ref{abstr-skew}, it suffices to show that for $(q, \bw)^T\in H$, $c \in \bR \setminus \{0\}$, there exists $(p,\bv)^T \in V$ for which $P(p,\bv)^T=(q,\bw)^T$, that is,
\begin{eqnarray}
\label{skew-cond}
-\nabla p +c \bv &=& \bw, \nonumber\\
-\nabla \cdot \bv +c p &=& q.
\end{eqnarray}

The first equation is equivalent to the requirement that for $\phi \in H^1_0(\Omega)$,
\[
-\langle \nabla \phi, \nabla p\rangle + c \langle\nabla \phi, \bv\rangle = \langle \nabla \phi, \bw\rangle.
\]
To satisfy the second equation in \eqref{skew-cond},  $\bv \in H^1_{\rm div}(\Omega)$ necessarily, and
\[
\langle\nabla \phi,\bv\rangle = -\langle \phi, \nabla \cdot \bv\rangle = \langle \phi, q - cp \rangle.
\]
Thus \eqref{skew-cond} implies that
\begin{equation}
\label{skew-cond-p}
-\langle \nabla \phi, \nabla p\rangle -c^2 \langle\phi, p\rangle = \langle\nabla \phi, \bw \rangle - c\langle \phi, q \rangle.
\end{equation}
The left-hand side of \eqref{skew-cond-p} defines a bounded and coercive (negative-definite) form on $H^1_0(\Omega)$ for any nonzero $c$, and the right hand side defines a continuous linear form on the same Hilbert space. The Lax-Milgram Theorem (\cite{Yosida}, III.7) implies the existence of a unique $p \in H^1_0(\Omega)$ for which \eqref{skew-cond-p} holds for every $\phi \in H^1_0(\Omega)$. 

With this choice of $p$, we are required to solve the second of the two conditions \eqref{skew-cond}. In fact, a solution $\bv_0 \in H^1(\Omega)^3 \subset H^1_{\rm div}(\Omega)$ exists satisfying 
\[
\|\bv_0\|_{H^1(\Omega)^3} \le C \| q - cp\|_{L^2(\Omega)}
\]
with a constant $C$ depending only on $\Omega$. For a proof, see \cite{BrennerScott:07}, Lemma 11.2.3, who also give references to other results for domains for polygonal, rather than smooth, boundaries.

Set $\bv_1 = c^{-1}\nabla p - \bv_0+ c^{-1}\bw \in L^2(\Omega)^3$. Then for any $\phi \in H^1_0(\Omega)$,
\begin{eqnarray}
\label{skew-corr}
-c\langle \nabla \phi, \bv_1 \rangle 
&= &\langle \nabla \phi, -\nabla p +c \bv_0 - \bw \rangle \nonumber\\
& = & \langle \nabla \phi, -\nabla p -\bw\rangle - c \langle \phi, \nabla \cdot \bv_0 \rangle \nonumber\\
& = & \langle \nabla \phi, -\nabla p -\bw\rangle +c\langle \phi, q-cp\rangle \nonumber\\
& = & 0,
\end{eqnarray}
thanks to \eqref{skew-cond-p}. That is, $\bv_1$ is divergence-free, in the sense of distributions, and in particular $\bv_1 \in H^1_{\rm div}(\Omega)$. Now set $\bv = \bv_0 + \bv_1$. Then the first of the two conditions in \eqref{skew-cond} is satisfied by $(p,\bv)$ thanks to the definiton of $\bv_1$, whereas the construction of $\bv_0$ and the divergence-free property of $\bv_1$ imply that the second condition is also satisfied.

Thus we have constructed a solution of \eqref{skew-cond} in $V = H^1_0(\Omega) \times H^1_{\rm div}(\Omega)$ for any $c \in \bR \setminus \{0\}$, whence we conclude that $P$ is skew-adjoint.
\end{proof}

\append{The Differential Case}\label{append:diff}
If the memory term (convolution operator $R$) is absent, then initial data determine solutions uniquely. In this section, we sketch the theory, parallel to that for causal solutions, which holds in this differential case. We assume throughout this section that $Q \equiv 0$.

\begin{corollary} Suppose that $u \in \llh$ is a weak solution of (\ref{formeq}). Then $u \in C^0(\bR,H)$.
\end{corollary}
\begin{proof} Choose $\phi \in C^{\infty}(\bR)$ so that $\phi(t) = 1$ for $t>1$ (say), and $\phi(t)=0$ for $t<-1$. Set $u_+=\phi u$, $u_-=(1-\phi)u$. It is straightforward to verify that $u_+$ is a causal solution of (\ref{formeq}) with $f$ replaced by $f+\phi'u \in \lh$, whence $u_+ \in C^0(\bR,H)$ according to Corollary\ref{cont}. Likewise $t \mapsto u_-(-t)$ is also a causal solution of (\ref{formeq}) with $D$, $B$ replaced by $-P$, $-B$, and $f$ replaced by $t \mapsto -f(-t) + \phi' u(-t)$, which also belongs to $\lh$. Thus $u_-$ is also continuous, but $u=u_++u_-$.
\end{proof}

\begin{corollary}\label{ivpenergy} Suppose that $u$ is a weak solution of \eqref{formeq}. Then for any $s \le t \in \bR$,
\[
E(t) \le E(s)+C \int_s^t \|f\|^2.
\]
\end{corollary}
\begin{proof}
The energy identity \eqref{eident} applies to weak solutions, causal or not. Take into account $R=0$, and use Gronwall's inequality, the boundedness of $B$, and the equivalence \eqref{enorm} of the energy with the norm in $H$.
\end{proof}

\begin{corollary} If $u_1$ and $u_2$ are weak solutions of \eqref{formeq} (with the same right-hand side $f \in \lh$), and $u_1(s) =u_2(s)$ for some $s \in \bR$, then $u_1 \equiv u_2$.
\end{corollary}
\begin{proof}
Set $u=u_1-u_2$: $u$ is a weak solution with right-hand side $f=0$, and $u(s)=0$. The result follows immediately from Corollary \ref{ivpenergy}.
\end{proof}

\begin{theorem} Suppose that $T_0 \in \bR$ and $u_0 \in H$. Then there exists a unique weak solution of \eqref{formeq} for which $u(T_0)=u_0$.
\end{theorem}
\begin{proof} A proof of this result is precisely analogous to the proof of Theorem \ref{existence}: the solution is approximated by a Galerkin procedure and the solution of systems of ordinary differential equations, and the energy in the error estimated (in this instance, via Corollary \ref{ivpenergy}). 
\end{proof}

Results precisely analogous to those established in the last section hold concerning regular dependence on the coefficient operators for weak solutions with specified initial data and right-hand side. We leave the reader to formulate these results, whose proofs are minor variants of those given above.

\end{document}